\documentclass[nonacm,screen]{acmart}
\pdfoutput=1

\usepackage{float}
\usepackage{color}
\usepackage{amsfonts,amsthm,amsmath,amssymb}
\usepackage{thmtools}

\newcounter{todocounter}
\newcommand{\todonum}{\stepcounter{todocounter}{(\thetodocounter)}}

\def\shownotes{0}   
\ifnum\shownotes=1
\newcommand{\authnote}[2]{{ $\ll$\textsf{\footnotesize \todonum\  #1 notes:  #2}$\gg$}}
\else
\newcommand{\authnote}[2]{}
\fi

\newcommand{\gs}[1]{{\color{magenta}\authnote{Gilad}{#1}}}
\newcommand{\IA}[1]{{\color{red}\authnote{Ittai}{#1}}}

\newcommand{\ignore}[1]{}
\newcommand{\hide}[1]{}

\makeatletter

\floatstyle{ruled}
\newfloat{algorithm}{tbp}{loa}
\providecommand{\algorithmname}{Algorithm}
\floatname{algorithm}{\protect\algorithmname}

\usepackage{algorithmic}

\makeatother
\begin{document}

\newtheorem{claim}{Claim}
\newtheorem{introtheorem}{Theorem}
\newtheorem{assumption}{Assumption}

\newcommand{\namedref}[2]{\hyperref[#2]{#1~\ref*{#2}}}
\newcommand{\sectionref}[1]{\namedref{Section}{#1}}
\newcommand{\appendixref}[1]{\namedref{Appendix}{#1}}
\newcommand{\subsectionref}[1]{\namedref{Subsection}{#1}}
\newcommand{\theoremref}[1]{\namedref{Theorem}{#1}}
\newcommand{\defref}[1]{\namedref{Definition}{#1}}
\newcommand{\figureref}[1]{\namedref{Figure}{#1}}
\newcommand{\claimref}[1]{\namedref{Claim}{#1}}
\newcommand{\lemmaref}[1]{\namedref{Lemma}{#1}}
\newcommand{\tableref}[1]{\namedref{Table}{#1}}
\newcommand{\corollaryref}[1]{\namedref{Corollary}{#1}}
\newcommand{\propertyref}[1]{\namedref{Property}{#1}}
\newcommand{\appref}[1]{\namedref{Appendix}{#1}}
\newcommand{\propref}[1]{\namedref{Proposition}{#1}}
\newcommand{\dagree}{{\em Detect/Agree\ }}

\newcommand{\F}{\mathcal{F}}
\renewcommand{\P}{\mathcal{P}}
\newcommand{\false}{\mathit{false}}
\newcommand{\true}{\mathit{true}}
\newcommand{\ack}{\mathit{ack}}
\newcommand{\info}{{\mbox{\it correct}}}
\newcommand\all{N}
\newcommand{\none}{\vec{\tau}'}
\newcommand{\best}{\mathrm{best}}
\newcommand{\punish}{\mathrm{punish}}
\newcommand{\majority}{\mathrm{majority}}
\newcommand\alone{\{i\}}
\newcommand{\secret}{\mathit{sec}}
\newcommand{\sigmact}{\sigma_{\mbox{\footnotesize\sc ct}}}
\newcommand{\Gammact}{\Gamma_{\!\mbox{\footnotesize\sc ct}}}
\newcommand{\vecsigmact}{\vec{\sigma}_{\mbox{\footnotesize\sc ct}}}
\newcommand{\uT}{u^{\mbox{\footnotesize\sc t}}}

\newcommand{\ceil}[1]{\left\lceil{#1}\right\rceil}
\newcommand{\pare}[1]{\left({#1}\right)}
\newcommand{\set}[1]{\left\{{#1}\right\}}
\newcommand{\ang}[1]{\left\langle {#1} \right\rangle}
\newcommand{\range}[2]{\set{{#1},\dots,{#2}}}
\setlength{\parskip}{.5ex}
\def\beginsmall#1{\vspace{-\parskip}\begin{#1}\itemsep-\parskip}
\def\endsmall#1{\end{#1}\vspace{-\parskip}}
\newcommand{\R}{\mathcal{R}}
\newenvironment{RETHM}[2]{\trivlist \item[\hskip 10pt\hskip\labelsep{\bf
#1\hskip 5pt\relax\ref{#2}.}]\it}{\endtrivlist}
\newcommand{\rethm}[1]{\begin{RETHM}{Theorem}{#1}}
\newcommand{\repro}[1]{\begin{RETHM}{Proposition}{#1}}
\newcommand{\relem}[1]{\begin{RETHM}{Lemma}{#1}}
\newcommand{\recor}[1]{\begin{RETHM}{Corollary}{#1}}

\newcommand{\erethm}{\end{RETHM}}
\newcommand{\erepro}{\end{RETHM}}
\newcommand{\erelem}{\end{RETHM}}
\newcommand{\erecor}{\end{RETHM}}

\def\appendix{\par
\section*{APPENDIX}
\setcounter{section}{0}
\setcounter{subsection}{0}
\def\thesection{\Alph{section}} }

\title{Revisiting Asynchronous Fault Tolerant Computation with Optimal Resilience}


\author{Ittai Abraham}
\affiliation{%
  \institution{VMware Research}
}

\author{Danny Dolev}
\affiliation{%
  \institution{Hebrew University}
}

\author{Gilad Stern}
\affiliation{%
  \institution{Hebrew University}
}

\begin{abstract}
The celebrated result of Fischer, Lynch and Paterson is the fundamental lower bound for asynchronous fault tolerant computation: any 1-crash resilient asynchronous agreement protocol must have some (possibly measure zero) probability of not terminating.
In 1994, Ben-Or, Kelmer and Rabin published a \textit{proof-sketch} of a lesser known lower bound for asynchronous fault tolerant computation with optimal resilience against a Byzantine adversary: if $n\le 4t$ then any t-resilient asynchronous verifiable secret sharing protocol must have some \textbf{non-zero} probability of not terminating.

Our main contribution is to revisit this lower bound and provide a rigorous and more general proof.
Our second contribution is to show how to avoid this lower bound. We provide a protocol with optimal resilience that is almost surely terminating for a \textit{strong common coin} functionality. Using this new primitive we provide an almost surely terminating protocol with optimal resilience for asynchronous Byzantine agreement that has a new \textit{fair validity} property. To the best of our knowledge this is the first   asynchronous Byzantine agreement with fair validity in the information theoretic setting.

\end{abstract}

\maketitle

\section{Introduction}

One of the most important models of distributed computing is the \emph{Asynchronous communication model}. Intuitively, this model captures the highest level of network un-reliability. It allows the adversary to delay each message arrival in an adaptive manner up to any finite amount. 
A basic question of distributed computing is:

\emph{Is there a fundamental limit to fault tolerant computation in the Asynchronous model?}

The celebrated Fischer, Lynch, and Merritt (FLP) \cite{FLP85} impossibility result from 1985 is perhaps the most well known such fundamental limitation. It states that reaching \textit{agreement}, even in the face of just one crash failure, is impossible for deterministic protocols. More formally, FLP~\cite{FLP85} prove that any protocol that solves Agreement in the asynchronous model that is resilient to at least one crash failure must have a \emph{non-terminating} execution. 
Thus, no protocol can solve Agreement in this model in finite time, but using randomization, it is possible to define a measure on the number of rounds and obtain protocols that have a finite \emph{expected} termination.
Given the FLP~\cite{FLP85} impossibility it is natural to ask: 

\textit{Is this potentially measure zero event of non-termination the only limitation for fault tolerant computation in the asynchronous model?}

In 1983,  Ben-Or, Canetti, and Goldreich  \cite{BCG93} initiated the study of secure multiparty computation in the asynchronous model. Their fundamental result is that the answer above is \textit{yes} when there are $n >4t$ servers and an adversary that can corrupt at most $t$ parties in a Byzantine (fully malicious) manner. They show that \emph{perfect} security with finite expected run time can be obtained for any functionality.

The BCG \cite{BCG93} work left open the domain of $3t<n \le 4t$ (with  $n=3t$ it is known that Byzantine agreement is impossible, see \cite{FLM}).
In 1993, Canetti and Rabin \cite{CR93} obtained a protocol for Asynchronous Byzantine Agreement with optimal resilience ($3t<n$). Their protocol had an "annoying property": the non-termination event has a \textbf{non-zero} probability measure. This problematic non-zero probability of non-termination came from their verifiable secret sharing protocol.
In 1994, Ben-Or, Kelmer and Rabin \cite{BKR94} addressed this problem. They provided an optimal resilience  asynchronous secure multiparty computation  protocol with the same "annoying property": the non-termination event has a \textbf{non-zero} probability measure \footnote{BCG \cite{BCG93}: "our protocol, as well as the verifiable secret sharing protocol of [CR93], have the following annoying property: the exponentially small error probability includes an exponentially small non-zero probability of not terminating. This should be contrasted with the asynchronous Byzantine Agreement problem where the randomized protocol terminates with probability 1."}. Moreover, BKR \cite{BKR94} claim that this is unavoidable. That is, if $n\le 4t$ then any t-resilient asynchronous verifiable secret sharing protocol $A$ must have some \textbf{non-zero} probability $q_A>0$ of not terminating. Unfortunately, BKR \cite{BKR94} only provided a \textit{proof-sketch} of the proof of this lower bound.

\subsection*{Our contributions: lower bounds on asynchronous verifiable secret sharing with optimal resilience}

25 years after the publication of the proof-stretch of BKR \cite{BKR94}, the main contribution of this paper is a rigorous proof of the lower bound theorem. We believe that our work will help provide clarity and better understanding of the asynchronous model and its impossibility results. In addition, our lower bound proof improves over the BKR proof-sketch in two important ways:

 One weakness of the BKR \cite{BKR94} proof-sketch is that its arguments only imply a lower bound for verifiable secret sharing schemes that have perfect hiding and binding properties. This raises a natural question: Can allowing some error probability in the AVSS scheme remove the need for a non-zero probability of non-termination? Our proof strengthens the BKR lower bound claim and proves this is not the case. We prove that even AVSS schemes with constant error must have a a non-zero probability of non-termination.

A second  weakness of the BKR \cite{BKR94} proof-sketch is that its arguments assume the share (and reconstruct) protocols terminate in a fixed (constant) number of rounds. VSS protocols whose share terminates with probability 1 have been shown to be useful in other contexts \cite{ADRH06}.  
Our proof strengthens the BKR \cite{BKR94} lower bound claim and proves that a non-zero probability of non-termination must occur even if the share and reconstruct protocols only terminate with probability 1.

\subsection*{Our contributions: upper bounds on strong common coin and asynchronous Byzantine agreement with fair validity}

What are the implications of this lower bound? Does it imply that \textit{all} optimal resilience secure computation must have a non-zero probability of non-termination? We know that this is not the case. In fact, Ben-Or \cite{Ben-Or} and Bracha \cite{B87} prove that Byzantine agreement has a measure zero probability of non-termination (it almost surely terminates) with optimal resilience. However the expected time of termination of these protocols is exponential.
The work of \cite{SVSS} shows that almost surely termination is possible even with a polynomial expected number of rounds. This is obtained using  a certain type of a \emph{weak common coin}  functionality that is also almost surely terminating.
This gap raises a natural question: 

\textit{Are there other functionalities (that are stronger than a weak coin, but weaker than verifiable secret sharing) that can be implemented in the asynchronous model for $n=3t+1$ that are almost surely terminating?}

Our first upper bound  contribution is to answer this question in the affirmative. We show that a certain type of a \emph{strong common coin} is possible to implement in an almost surely terminating manner. The difference between a weak common coin and a strong common coin is that in a strong common coin protocol, all parties output the same value while in a weak common coin, with constant probability, different parties may output different values for the coin.

What is the advantage of a strong common coin over a weak common coin? With a strong common coin we know how to obtain asynchronous Byzantine agreement with a  \textit{fair validity} property. We do not know how to obtain this validity property with a weak coin.

\IA{say more about fair validity...}

Our second upper bound contribution is a Byzantine Agreement protocol in the Asynchronous model for $n=3t+1$  with \textit{fair validity} that is almost surely terminating. To the best of our knowledge this is the first Asynchronous Byzantine Agreement protocol  with \textit{fair validity} in the information theoretic setting.

\section{Lower Bound}\label{sec:main}

\begin{definition}
A Byzantine AVSS protocol, comprised of a pair of protocols $\left(S,R\right)$\footnote{$S$ is the protocol for sharing a secret and $R$ is the protocol for reconstructing it.},
has a designated dealer called $D$, which receives a secret $s$
from a finite field $\F$ as input. For $\epsilon>0$, such a protocol
is called an almost-surely terminating $\left(1-\epsilon\right)$-correct $t$-resilient
AVSS protocol if the three following properties hold for every adversary
controlling $t$ parties at most, and any message scheduling:
\begin{enumerate}
\item \textbf{Termination}:
\begin{enumerate}
\item If the dealer is nonfaulty and all nonfaulty parties participate in protocol
$S$, then each nonfaulty party will almost-surely eventually complete protocol $S$.
\item If some nonfaulty party completed protocol $S$, then each nonfaulty party that participates in $S$
will almost-surely eventually complete protocol $S$.
\item If all of the nonfaulty parties finished protocol $S$ and began protocol
$R$, they will all almost-surely complete protocol $R$.
\end{enumerate}
\item \textbf{Correctness}.
Once the first nonfaulty party has completed protocol $S$, there exists
some value $r\in \F$ such that with a probability of at least $\left(1-\epsilon\right)$:
\begin{enumerate}
\item If the dealer is nonfaulty, $r=s$.
\item Every nonfaulty party that completes protocol $R$ outputs the value
$r$.
\end{enumerate}
\item \textbf{Secrecy}.
If the dealer is nonfaulty, and no honest party has began protocol $R$,
no adversary can gain any information about $s$. 
More precisely, denote $V^s$ to be the adversary's view of an execution of $S$ with a nonfaulty dealer sharing $s$ before some nonfaulty party calls protocol $R$. 
If the dealer is nonfaulty, then for any given adversary
and message scheduling, the distribution of $V^s$
is the same for all possible secrets $s$. 
\end{enumerate}
\end{definition}

If some party almost-surely completes the protocol, it must complete the protocol in finite time with probability 1.
This also means that for every $\epsilon>0$ there exists some number $N\in\mathbb{N}$ such that the probability that the party exchanges more than $N$ messages with all parties during protocol $S$ is less than $\epsilon$.
It is important to note that those values might need to be adjusted based on the adversary and scheduling as well.
Similarly, if all parties almost-surely terminate, for every $\epsilon>0$ there exists some $N\in\mathbb{N}$ such that the probability that there exists a nonfaulty party who exchanges more than $N$ messages is no greater than $\epsilon$.
The main result shown in this section is proving the following theorem:
\begin{theorem}\label{thm:impossibility} For any $\epsilon>0$ and $n\leq 4t$ there is no terminating $\left(\frac{1}{2}+\epsilon\right)$-correct $t$-resilient Byzantine AVSS protocol $\left(S,R\right)$.
\end{theorem}

Let $n=4$, $t=1$ and assume a binary secret $s\in\left\{ 0,1\right\}$.
Using standard methods, this result can be expanded to any $4t\geq n\geq 3t+1$, and to a multivalued secret.
Let the parties be $A,B,C,D$, and let $D$ be the dealer. 

By way of contradiction, assume the parties run an almost-surely terminating $\left(\frac{1}{2}+\epsilon\right)$-correct $t$-resilient Byzantine AVSS protocol.
The theorem is proven using two main claims. 
The first claim describes possible malicious behaviour by a faulty dealer during protocol $S$. 
The second claim describes possible malicious behaviour by another party in protocol $R$.

Before we state the first claim we define a distribution of views where the system is synchronous, the dealer $D$ and parties $A,B$ are nonfualty and party $C$ has crashed.
In such a setting, from the Termination property of $AVSS$ all nonfaulty parties almost-surely complete protocol $S$.
Set some $1>\epsilon'> 0$ and let $N\in\mathbb{N}$ be a number such that if processors $A,B,D$ participate in protocol $S$ in the setting described above, the probability that one of them runs for more than $N$ rounds throughout protocol $S$ is smaller than $\epsilon'$.
From the Termination property of the protocol, and since the setting is synchronous, such a value $N$ must exist.
Define $long$ to be the event in which either $A,B$ or $D$ run for longer than $N$ rounds throughout protocol $S$, and $\overline{long}$ to be its complement.

\begin{definition}
For every $s\in\left\{0,1\right\},P\in\left\{A,B\right\}$, let $\pi_{s,P}$ be the distribution of $P$'s view when a nonfaulty $D$ shares the value $s$, conditioned upon the event $\overline{long}$.
\end{definition}

In the attack described in the first claim, the dealer, $D$ has a non zero probability of causing parties $A$ and $B$ to complete protocol $S$, but the conditional distribution of party $A$'s view is $\pi_{0,A}$ while the conditional distribution of party $B$'s view is $\pi_{1,B}$.

\begin{claim}\label{claim:firstattack}
A faulty dealer $D$ has some strategy such that with some nonzero probability the following event holds:
\begin{enumerate}
    \item Parties $A$ and $B$ complete protocol $S$.
    \item The conditional distribution of party $A$'s view is $\pi_{0,A}$.
    \item The conditional distribution of party $B$'s view is $\pi_{1,B}$.
\end{enumerate}
\end{claim}
Intuitively, the dealer will try to make party $A$ and party $B$ complete protocol $S$ while seeing contradictory worldviews. Conditioned on this nonzero probability event the following happens.
On the one hand, in $A$'s view, the execution of $S$ looks like one in which $D$ shared the value $0$ and $C$ was faulty and silent (corresponding to the distribution $\pi_{0,A}$).
On the other hand, in $B$'s view the execution of $S$ looks like one in which $D$ shared the value $1$ and $C$ was faulty and silent (corresponding to the distribution $\pi_{1,B}$).
After $A$ and $B$ complete protocol $S$, $C$ will start participating in the protocol, and complete protocol $S$ as well. 
Then all three parties will participate in protocol $R$, and eventually complete it and output some value $r$. This value $r$ will be used in the next attack in the second claim.

The crux of this attack is that parties $A$ and $B$ should not know whether $0$ or $1$ is shared during $S$ in a run in which $D$ is nonfaulty because of the Secrecy property.
Leveraging this ambiguity, the faulty dealer tries to make them complete $S$ with incompatible views, which they will have to resolve during protocol $R$.

Afterwards, in order to prove the main result, we prove the following claim for every $1>\epsilon'>0$ (arbitrarily close to 0):
\begin{claim}\label{claim:secondattack}
Without loss of generality, the adversary has a strategy controlling party $B$ such that when a nonfaulty $D$ shares the value $0$, with probability at least $1-\epsilon'$:
\begin{itemize}
    \item $A$'s view during protocol $S$ is distributed according to the distribution $\pi_{0,A}$,
    \item $C$ outputs $0$ at the end of protocol $R$ with probability $\frac{1}{2}$ or less.
\end{itemize}
\end{claim}
The value $0$ in the claim is used when in the previous attack party $C$ outputs the value $r=0$ with probability $\frac{1}{2}$ or less. The claim is "without loss of generality", in the sense that if $C$ outputs the value $r=1$ with probability $\frac{1}{2}$ or less then we switch $A$ with $B$ as well as $0$ with $1$. 

In this attack, $B$ acts normally throughout $S$, and then throughout $R$ acts as if the attack in claim~\ref{claim:firstattack} took place.
As opposed to the previous attack which can succeed with some nonzero probability, $B$'s attack will succeed (i.e. it will look as if the attack in claim ~\ref{claim:firstattack} took place) with a probability arbitrarily close to $1$.
This means that the event that C outputs 0 with probability $\frac{1}{2}$ or less can occur with a probability arbitrarily close to 1.
Note that if the dealer shares $0$, then every nonfaulty party should output $0$ with probability $\frac{1}{2}+\epsilon$ or greater, and thus $C$ can only fail to output $0$ with probability $\frac{1}{2}-\epsilon$.
Therefore, for a small enough $\epsilon'$ we  reach a contradiction. This concludes the proof for Theorem \ref{thm:impossibility}.

Several random variables are defined in order to prove the aforementioned claims.
Technically, the distribution of the random variables could depend on the dealer, the adversary and on the message scheduling. 
Throughout all of the analysis these factors are strictly defined, and are therefore omitted from the definitions of the random variables.

Let $M^{s}_{XY}$ be the distribution of messages exchanged between
party $X$ and party $Y$ during the sharing protocol $S$
if the network is synchronous, party $D$ is a nonfaulty dealer sharing the value $s$, party $C$ is faulty and silent, and no nonfaulty party calls protocol $R$.
Let $R^{s}_{X}$ be the distribution of the internal randomness of party
$X$ throughout the sharing protocol in the described setting.
Let $V_{X}^{s}$  be the distribution of party $X$'s view of the share
phase in the setting described above.
For every pair of parties $X,Y$ and value $s$ let $r^s_X\sim R^s_X$ be a random variable describing $X$'s randomness in a given run, $m^s_{XY}\sim M^s_{XY}$ be a random variable describing the messages between parties $X$ and $Y$ with a non faulty dealer sharing $s$ in the described setting, and $v^s_X\sim V^s_X$ be a random variable describing $X$'s view in the run.
Note that $r^s_{X}$ is necessarily part of $v_{X}^{s}$ in some way, as well as $m_{XY}^{s}$
for every party $Y$. 

For any distribution $X$, $x\in X$ means that $x$ has a nonzero probability under $X$.
Party $X$'s view $v_X$ is \textit{consistent} with $s$ if $v_{X}\in V_{X}^{s}$. 
Similarly, a set of messages $m_{XY}$ exchanged between party $X$ and party $Y$ is \textit{consistent}
with the secret $s$ if $m_{XY}\in M_{XY}^{s}$.

For the first part of the analysis, assume $D$ is corrupted by the
adversary.
The overarching goal is to prove claim \ref{claim:firstattack} while processor $A$ sees a view consistent with $D$ sharing the value $0$ and processor $B$ sees a view consistent with $D$ sharing the value $1$.
Intuitively, from the Secrecy property, neither party $A$ nor party $B$ should be able to tell which value was shared throughout protocol $S$, and thus $D$ could send messages in that manner and neither party would notice.

The adversary's strategy is as follows:
Party $D$ samples $s_{A}\leftarrow R^0_{A}|\overline{long}$ and $s_{AB}\leftarrow M^0_{AB}|r^0_A=s_A,\overline{long}$ and then $s_B\leftarrow R^1_B|m^1_{AB}=s_{AB},\overline{long}$.
Afterwards it samples $s_{AD}\leftarrow M_{AD}^{0}|m_{AB}^{0}=s_{AB},r^0_{A}=s_{A},\overline{long}$
and $m_{BD}\leftarrow M_{BD}^{1}|m_{AB}^{1}=s_{AB},r^1_{B}=s_{B},\overline{long}$.
The random variables $s_A,s_B$ are $D$'s guesses of $A$ and $B$'s randomness, and the variables $s_{AB},s_{AD},s_{BD}$ are the messages $D$ predicts will be sent.
Define party $X$'s randomness throughout this run to be $r_X$ and the messages exchanged between parties $X$ and $Y$ in this run to be $m_{XY}$.
Finally, define the event $G$, in which $s_{A}=r_{A},s_{B}=r_{B}$.

Before showing this behaviour can be used as part of claim \ref{claim:firstattack}, we need to show that these distributions are well-defined and samplable. 
Note that the setting is entirely synchronous.
For simplicity, assume that the number of bits in a message and the amount of randomness needed in each round are bounded.
In addition, assume that every party sends some message indicating that it completes protocol $S$.
These assumptions are used in order to simplify the proof of the next lemma.
\footnote{In order to prove the general case, the dealer can simulate the entire run for parties $A,B,D$ round-by-round twice, once sharing the value $0$ and once sharing the value $1$. The dealer will only accept pairs of runs in which the messages exchanged between processors $A$ and $B$ are the same.
Proving that there must exist such a pair of runs requires proving a lemma similar to the following lemma without conditioning upon the event $\overline{long}$.
Note that since the rounds almost-surely terminate, the sampling process will also terminate with probability $1$.
This will result in slight differences in the attacks and proofs, but with very similar techniques and ideas.
The main difference is that all of the sampled probabilities will not be conditioned upon the event $\overline{long}$.}

\begin{lemma}\label{lem:indlong} For every values $m'_{AD},m'_{AB},r'_A$ such that 
$\Pr [ m^0_{AD} = m'_{AD} , m^0_{AB} = m'_{AB} , r^0_A = r'_A,\overline{long}]\neq 0$, $
\Pr[ m_{AD}^{0} = m'_{AD} , m_{AB}^{0} = m'_{AB} , r^0_A = r'_A | \overline{long}]{=}
\Pr [m_{AD}^{1} = m'_{AD} , m_{AB}^{1} = m'_{AB} , r^1_A = r'_A | \overline{long}].
$
\end{lemma}
\begin{proof}
First note that from the Termination property of $AVSS$,
if all nonfaulty parties participate in protocol $S$ they will all complete it. 
Observe a scheduling in which the communication between parties $A,B$ and $D$ is synchronous, party $C$ is silent throughout all of protocol $S$ and no nonfaulty party calls protocol $R$ at all.
In this case, since no nonfaulty party calls protocol $R$,
the Secrecy property must hold at the time the parties complete protocol $S$.

Seeking a contradiction, assume the lemma doesn't hold and show a violation of the Secrecy property.
In that case, observe the scenario in which the adversary controls party $A$ and the nonfaulty dealer shares the value $0$.
Party $A$ acts like a nonfaulty party would throughout all of protocol $S$.\IA{not clear} \gs{Is this any better?}
From $D$ and $B$'s point of view, party $A$ acts as a nonfaulty party and $C$ acts as a faulty party which doesn't send any messages. 
Since they cannot distinguish between the scenarios, the messages must be distributed according to the distributions $M_{AD}^{0},M_{AB}^{0}$ and $A$'s randomness must be distributed according to $R^0_A$. 
Since
$\Pr\left[m_{AD}^{0}=m'_{AD},m_{AB}^{0}=m'_{AB},r^0_A=r'_A,\overline{long}\right]\neq0$,
there is a nonzero probability that
parties $A,B,D$ complete protocol $S$ with $A$ having exchanged those messages in fewer than $N$ rounds.
All processors also know when other processors complete the protocol because they send a message indicating they completed protocol $S$.
No nonfaulty party called $R$ yet,
and $\Pr [ m_{AD}^{0} = m'_{AD} , m_{AB}^{0} = m'_{AB} , r^0_A =r'_A | \overline{long}] \neq \Pr [m_{AD}^{1} = m'_{AD} , m_{AB}^{1} = m'_{AB} , r^1_A = r'_A | \overline{long}]$.
In other words, if party $A$ acted in the exact same way, and $D$ were sharing the value $s=1$ instead, $A$ would have had a different probability of seeing these values if the event $\overline{long}$ took place.
Since $A$'s random values and the messages it exchanges are a part of its view, as well as its knowledge as to whether the event $long$ happened,
this means that this adversary's view in the case $s=0$ is distributed differently than it would be in the case $s=1$ reaching a contradiction. 
For completeness, all messages to and from $C$ can be sent and received after
parties $A,B$ and $D$ complete protocol $S$ in order for the scheduling to be valid.
\end{proof}
The probabilities are equal for every nonzero-probability event in the case that the dealer is sharing the value $0$.
Since both must be probability spaces $\Pr\left[m_{AD}^{0}=m'_{AD},m_{AB}^{0}=m'_{AB},r^0_A=r'_A|\overline{long}\right]=0$ if and only if it is also true that $\Pr[m_{AD}^{1}=m'_{AD},m_{AB}^{1}=m'_{AB},r^1_A=r'_A|\overline{long}]=0$, and thus the distributions must be identical. 
In addition, the exact same arguments can be made for $B$ instead or if party $D$ shared the value $s=1$. 

A direct corollary is that any of the marginal and conditional probabilities are also the same. For example:

\begin{restatable}{corollary}{indmarginallong}
\label{col:indmarginallong}
For every $m'_{AD}\in M^0_{AD}|\overline{long}$,  $\Pr[m^0_{AD} = m'_{AD} | \overline{long}] = \Pr[m^1_{AD}=m'_{AD} | \overline{long}]$.
Also, for every $m'_{AB}\in M^0_{AB}|\overline{long}$,
$\Pr[m^0_{AD}=m_{AD}|m^0_{AB}=m'_{AB},\overline{long}]=
\Pr[m^1_{AD}=m_{AD}|m^1_{AB}=m'_{AB},\overline{long}]$.
\end{restatable}

\begin{restatable}{corollary}{samplable}
\label{col:samplable}
The values sampled by $D$ in the described attack are sampled from well-defined, samplable distributions.
\end{restatable}
The proofs of the corollaries are provided in the appendix.

In general the strategy from this point on is to prove that if any party's view is consistent with some secret it must complete protocol $S$.
The next step is to show that if event $G$ occurs, $A$ and $B$'s view must be consistent with some secret $s$ and the event $\overline{long}$ must take place.
In the runs by the random variables party $C$ is faulty and doesn't send any messages, and thus we prove that there are runs in which parties $A$ and $B$ must complete protocol $S$ even without receiving any messages from $C$.

\begin{lemma}\label{lem:termination} If party $A$'s view is consistent with some secret $s\in\left\{0,1\right\}$ then it must almost-surely complete the protocol, even without receiving any messages from party $C$.
\end{lemma}
\begin{proof}
Since party $A$'s view is consistent with the secret $s$, $A$ could have had this exact view with a nonfaulty dealer $D$ sharing $s$ if $C$ were faulty and silent. 
Since in that run $A$ and $B$ are nonfaulty and $D$ is a nonfaulty dealer, 
from the Termination property of $AVSS$, $A$ must almost-surely complete protocol $S$ in that run. 
$A$ can't tell the difference between its view in this run and the view in which $C$ was faulty, 
and thus $A$ must complete protocol $S$ if its view is merely consistent with $s$.
Party $C$'s messages won't be infinitely delayed in this run because the probability that party $A$ has an infinitely large view is $0$.
\end{proof}

This exact argument can be made for party $B$ as well.
We now turn to show that if $D$ acts according to the described strategy, there is a nonzero probability that $A$ and $B$ will complete protocol $S$ with the desired distribution of views, conditioned upon the event $\overline{long}$.

\begin{restatable}{lemma}{distribution}
\label{lem:distribution} If the dealer is faulty, $s_{A}=r_{A},s_{B}=r_{B}$,
$D$ exchanges the messages $s_{AD}$ and $s_{BD}$ with parties $A$ and $B$ respectively, 
parties $A$ and $B$ exchange the messages $s_{AB}$ between them,
and the scheduling is as described above, then party $A$'s view is distributed according to $V_{A}^{0}|\overline{long}$ and party $B$'s view is distributed according to $V_{B}^{1}|\overline{long}$.
\end{restatable}
The proof of this lemma can be found in the appendix.

\begin{lemma}\label{lem:firstattacksuccessful} If the dealer is faulty, $s_{A}=r_{A},s_{B}=r_{B}$, 
$D$ sends messages to parties $A$ and $B$ according to $s_{AD}$ and $s_{BD}$,
and the scheduling is as described above, 
then parties $A$ and $B$ complete protocol $S$ having exchanged $s_{AD},s_{BD}$ with $D$ respectively and $s_{AB}$ between them.
\end{lemma}
\begin{proof}
If party $D$ correctly guesses $s_A=r_A,s_B=r_B$, then the messages $A$ and $B$ exchange with each other and with $D$ in response to each of $D$'s messages are going to become entirely deterministic and dictated only by $D$'s messages.
This means that since $D$'s messages are always going to be consistent with the sampled values $s_{AD},s_{BD}$, parties $A$ and $B$ are going to send the appropriate responses to $D$, as well as exchange the messages $s_{AB}$ sampled by $D$ between them.
In that case, from lemma \ref{lem:distribution} party $A$'s view is distributed according to $V_A^0|\overline{long}$ and party $B$'s view is distributed according to $V_B^1|\overline{long}$.
This means that party $A$'s view is consistent with $s=0$ and party $B$'s view is consistent with $s=1$.
From lemma \ref{lem:termination} this means that parties $A$ and $B$ almost-surely complete protocol $S$ in finite time.
\end{proof}
\gs{This feels extremely hand-wavey. I don't know how to make this more rigorous.}

In order for the scheduling to be valid, 
once parties $A$ and $B$ complete protocol $S$,
all messages to and from party $C$ are instantly delivered.
Note that party $D$ hasn't sent any messages to party $C$.
There is a nonzero probability of $s_A=r_A,s_B=r_B$, and thus claim~\ref{claim:firstattack} is proven by combining lemma~\ref{lem:distribution} and lemma~\ref{lem:firstattacksuccessful}

Now observe the following behaviour and scheduling after protocol $S$: processor $D$ now stays silent throughout all of protocol $R$, 
and all of the messages to and from parties $A,B$ and $C$ are synchronously delivered. 
Since all nonfaulty parties participate in protocol $S$, and some nonfautly party completed protocol $S$, 
all nonfaulty parties almost-surely complete it as well.
Similarly, since all nonfaulty parties completed protocol $S$ and participate in protocol $R$, they all almost-surely complete it as well.
Define $O_{C}$ to be the random variable describing the output of party $C$ during these runs, conditioned upon the event $G$. 
In other words, only observe the runs in which party $D$ correctly guessed the other parties' randomness.
Now, it is either the case that $\Pr\left[O_C=0\right]\leq \frac{1}{2}$ or the case that $\Pr\left[O_C=1\right]\leq \frac{1}{2}$.

The rest of the section proves claim~\ref{claim:secondattack} by describing attacks in which the adversary controls either party $A$ or party $B$ and simulates the previous adversary's behaviour conditioned upon the event $G$.
It is possible for the adversary to simulate that event with probability $1-\epsilon'$ even though the event has a negligible probability of occurring in the original attack, gaining a significant advantage.
First assume that $\Pr\left[O_C=0\right]\leq \frac{1}{2}$.
In that case, the adversary can control party $B$ with some specific scheduling in such a way that if a nonfaulty dealer shares the value $0$ and the event $\overline{long}$ takes place, party $A$'s view throughout the protocol must be distributed according to $V^0_A|\overline{long}$.
Party $B$ also acts in a way similar to the way it would have acted in the previous attack. 
This means that all parties act in the same way they would have acted in the original attack, and thus party $C$ outputs $0$ with probability $\frac{1}{2}$ or less if the event $\overline{long}$ takes place.
Since the event $\overline{long}$ takes place with at least a probability of $1-\epsilon'$, this proves the claim.  

\begin{lemma}\label{lem:Battack}
If $\Pr\left[O_C=0\right]\leq\frac{1}{2}$, there exist an adversary controlling party $B$ and a scheduling such that with probability $1-\epsilon'$ or more the following things hold when a nonfaulty dealer $D$ shares the value $0$:
\begin{itemize}
    \item party $A$'s view during protocol $S$ is distributed according to $V^0_A|\overline{long}$,
    \item party $C$ outputs $0$ at the end of protocol $R$ with probability $\frac{1}{2}$ or less.
\end{itemize}
\end{lemma}
\begin{proof}
The scheduling is described only in case the dealer shares the value $s=0$ and no party runs for longer than $N$ round.
Any other valid scheduling can take place if those conditions don't hold.
The adversary takes control of party $B$, and makes it act as a nonfaulty party would throughout all of protocol $S$.
All communications between parties $A,B$ and $D$ are synchronous throughout protocol $S$.
In addition, all messages to and from $C$ are delayed until parties $A,B$ and $D$ complete protocol $S$.
Since party $B$ is acting as a nonfaulty party would, parties $A$ and $D$ can't tell the difference between this run and a run in which party $C$ is faulty and silent.
As discussed above, in this situation parties $A,B$ and $D$ must complete protocol $S$.

Let $\hat{m}_{XY}$ be the messages party $X$ exchanged with party $Y$ throughout protocol $S$, and let $\hat{r}_X$ be party $X$'s randomness throughout the protocol.
After completing protocol $S$, party $B$ simulates all runs in which $\overline{long}$ takes place when a nonfaulty dealer shares the value $1$.
If there is no such run in which the messages $\hat{m}_{AB}$ are exchanged between parties $A$ and $B$, party $B$ acts as a nonfaulty processor throughout protocol $R$.
Otherwise, party $B$ samples some random values $\hat{s}_B\leftarrow R^1_B|m^1_{AB}=\hat{m}_{AB},\overline{long}$ and some messages $\hat{s}_{BD}\leftarrow M^1_{BD}|m^1_{AB}=\hat{m}_{AB},r^1_B=\hat{s}_B,\overline{long}$.
Note that clearly $\Pr\left[m^0_{AB}=\hat{m}_{AB}|\overline{long}\right]\neq 0$, and thus also $\Pr\left[m^1_{AB}=\hat{m}_{AB}|\overline{long}\right]\neq 0$ from corollary \ref{col:indmarginallong}.
This means that the above distributions are well-defined.
From this point on, party $B$ acts as a nonfaulty party would act with a view consisting of $\hat{m}_{AB},\hat{s}_{BD},\hat{s}_B$.
After parties $A$ and $B$ complete protocol $S$ all messages between parties $A,B$ and $C$, including the messages previously sent, are synchronously delivered.
All messages to and from party $D$ are delayed until the rest of the parties complete protocol $R$.
It is important to note that in this scheduling, all the messages party $D$ sent to party $C$ throughout protocol $S$ are also delayed until after all nonfaulty parties complete protocol $R$.

Recall that $m_{XY}$ is defined as the messages exchange by parties $X$ and $Y$ and $r_X$ is defined as $X$'s randomness throughout the attack described in claim \ref{claim:firstattack}. 
Now observe a snapshot of the values party $A$ saw throughout protocol $S$ and the values party $B$ claims it saw throughout the protocol.
For any values $r'_A,r'_B,m'_{AB},m'_{BD},m'_{AD}$ such that $\Pr[m_{AB}=m'_{AB},m_{AD}=m'_{AD},m_{BD}=m'_{BD},r_A=r'_A,r_B=r'_B|G]\neq 0$
the following also holds:

\begin{align*}
    &Pr\left[m_{AB}{=}m'_{AB},m_{AD}=m'_{AD},m_{BD}=m'_{BD},r_A=r'_A,r_B=r'_B|G\right]
\end{align*}
\begin{align*}
    &=\Pr\left[s_{AB}{=}m'_{AB},s_{AD}=m'_{AD},s_{BD}=m'_{BD},r_A=r'_A,r_B=r'_B|G\right] \\
    &=\Pr\left[s_{AB}=m'_{AB},s_{AD}=m'_{AD},s_A=r'_A|G\right] \\
    &\cdot\Pr\left[s_B=r'_B|s_{AB}=m'_{AB},s_{AD}=m'_{AD},s_A=r'_A\right]\\
    &\cdot\Pr\left[s_{BD}=m'_{BD}|s_{AB}=m'_{AB},s_{AD}=m'_{AD},s_A=r'_A,s_B=r'_B\right] \\
    &=\Pr\left[m^0_{AB}=m'_{AB},m^0_{AD}=m'_{AD},r^0_A=r'_A|\overline{long}\right] \\
    &\cdot\Pr\left[r^1_B=r'_B|m^1_{AB}=m'_{AB},\overline{long}\right]  \\
    &\cdot\Pr\left[m^1_{BD}=m'_{BD}|m^1_{AB}=m'_{AB},r^1_B=r'_B,\overline{long}\right]
\end{align*}
Where the last equality stems from several facts.
From lemma~\ref{lem:distribution} $\Pr[s_{AB}=m'_{AB},s_{AD}=m'_{AD},r_A=r'_A|G]=\Pr[m^0_{AB}{=}m'_{AB},m^0_{AD}=m'_{AD},r^0_A=r'_A|\overline{long}]$.
From the way the random variable $s_B$ is sampled, given $s_{AB}$ the variable $s_B$ is independent of the variables $s_{AD},s_A$.
Now, $\Pr[s_B=r'_B|s_{AB}=m'_{AB}]=\Pr[r^1_B=r'_B|m^1_{AB}=m'_{AB},\overline{long}]$ from the definition of $s_B$.
A similar argument can be made for the final expression.

On the other hand, note that since parties $A,B$ and $D$ are acting as nonfaulty parties throughout $S$, their actions are distributed identically to the setting in which $C$ is faulty and silent.
In this setting, the event $\overline{long}$ takes place with probability $1-\epsilon'$ at the very least.
Note that if this event takes place, then processor $B$ sees that the messages $\hat{m}_{AB}$ can be exchanged in some run in which the event $\overline{long}$ takes place, and thus sample some values.
Therefore, conditioned upon the event $\overline{long}$:
\begin{align*}
    &\Pr\left[\hat{m}_{AB}{=}m'_{AB},\hat{m}_{AD}=m'_{AD},\hat{s}_{BD}=m'_{BD},\hat{r}_A=r'_A,\hat{s}_B=r'_B|\overline{long}\right] \\
    &=\Pr\left[\hat{m}_{AB}=m'_{AB},\hat{m}_{AD}=m'_{AD},\hat{r}_A=r'_A|\overline{long}\right] \\
    &\cdot\Pr\left[\hat{s}_{B}=r'_B|\hat{m}_{AB}=m'_{AB},\hat{m}_{AD}=m'_{AD},\hat{r}_A=r'_A,\overline{long}\right]  \\
    &\cdot\Pr\left[\hat{m}_{BD}{=}m'_{BD}|\hat{m}_{AB}{=}m'_{AB},\hat{m}_{AD}{=}m'_{AD},\hat{r}_A{=}r'_A,\hat{s}_B=r'_B,\overline{long}\right]  \\
    &=\Pr\left[m^0_{AB}=m'_{AB},m^0_{AD}=m'_{AD},r^0_A=r'_A|\overline{long}\right] \\
    &\cdot\Pr\left[r^1_B=r'_B|m^1_{AB}=m'_{AB},\overline{long}\right]  \\
    &\cdot\Pr\left[m^1_{BD}=m'_{BD}|m^1_{AB}=m'_{AB},r^1_B=r'_B,\overline{long}\right]
\end{align*}
Where the last equality stems from similar arguments.
First of all, note that from $A$'s point of view, party $C$ is acting like a faulty party which is staying silent throughout protocol $S$ and parties $B,D$ are acting as nonfaulty parties with $D$ sharing the value $0$.
Therefore, $\Pr[\hat{m}_{AB}=m'_{AB},\hat{m}_{AD}=m'_{AD},\hat{r}_A=r'_A|\overline{long}]=\Pr[m^0_{AB}=m'_{AB},m^0_{AD}=m'_{AD},r^0_A=r_A|\overline{long}]$.
This also means that if the event $\overline{long}$ takes place, party $A$'s view is distributed according to $V^0_A|\overline{long}$.
From the way $\hat{s}_B$ is sampled, given $\hat{m}_{AB}$, the random variable $\hat{s}_B$ is entirely independent of $\hat{m}_{AD},\hat{r}_A$.
Taking that fact into consideration, and looking at the definition of $\hat{s}_B$, $\Pr[\hat{s}_B=r'_B|\hat{m}_{AB}=m'_{AB},\hat{m}_{AD}=m'_{AD},\hat{r}_A=r'_A,\overline{long}]=\Pr[r^1_B=r'_B|m^1_{AB}=m'_{AB},\overline{long}]$.
A similar argument can be made for $\hat{s}_{BD}$.

Party $B$'s behaviour is identical to the behaviour it would have in the attack described in the first part, and party $A$'s view is identical to that view as well.
From this point on, protocol $R$ is run in the exact same way, and neither party $A$ nor party $C$ can tell the difference between the runs in which party $B$ was faulty and event $\overline{long}$ occurred, and the runs in which party $D$ was faulty, given that event $G$ occurred.
First of all note that in the previous attack, parties $A$ and $C$ output some value before receiving messages from party $D$ during protocol $R$, and thus must do so in this scenario as well.
In order for the scheduling to be valid, all of the messages to and from party $D$ are received some finite time after party $A$ and party $C$ output a value.
The distribution of $A$ and $C$'s views in the beginning of protocol $R$ is identical to the distribution of their views in the previous attack.
Furthermore, party $B$'s actions are defined by the view it is simulating in the beginning of protocol $R$ as well.
Since parties $A,B$ and $C$'s actions are determined by their view at any point in time, the distribution of their views throughout the rest of protocol $R$ is identical in both runs as well, and thus the distributions of their outputs must be the same as well. 
Therefore if event $\overline{long}$ occurred, the probability that party $C$ outputs $0$  $\frac{1}{2}$ or less.
Note that event $\overline{long}$ takes place with probability $1-\epsilon'$ or more, which completes the lemma.

\end{proof}

Now assume that $\Pr\left[O_C=1\right]\leq\frac{1}{2}$. In that case:

\begin{restatable}{lemma}{Aattack}
\label{lem:Aattack}
If $\Pr\left[O_C=1\right]\leq\frac{1}{2}$, there exist an adversary controlling party $A$ and a scheduling such that with probability $1-\epsilon'$ or more the following things hold when a nonfaulty dealer $D$ shares the value $1$:
\begin{itemize}
    \item party $B$'s view during protocol $S$ is distributed according to $V^1_B|\overline{long}$,
    \item party $C$ outputs $1$ at the end of protocol $R$ with probability $\frac{1}{2}$ or less.
\end{itemize}
\end{restatable}
\begin{proof}
The proof of this lemma is extremely similar to the proof of lemma~\ref{lem:Battack}, and is thus provided in the appendix.
\end{proof}

If $\Pr\left[O_C=0\right]\leq\frac{1}{2}$, lemma \ref{lem:Battack} shows that claim \ref{claim:secondattack} holds when $s=0$ with the adversary controlling $B$.
On the other hand, if $\Pr\left[O_C=1\right]\leq\frac{1}{2}$, lemma \ref{lem:Aattack} shows that claim \ref{claim:secondattack} holds when $s=1$ with the adversary controlling $A$.
Since either $\Pr\left[O_C=0\right]\leq\frac{1}{2}$ or $\Pr\left[O_C=1|G\right]\leq\frac{1}{2}$, claim \ref{claim:secondattack} must hold.
Now, assume w.l.o.g that $\Pr\left[O_C=0\right]\leq \frac{1}{2}$.
Then, if a nonfaulty dealer $D$ shares the value $0$, an adversary has a strategy controlling $B$ such that for any $1>\epsilon'>0$ party $C$ outputs $0$ with probability no greater than $\frac{1}{2}$ if an event occurs with probability $1-\epsilon'$ or more.
In addition to that if that event doesn't occur (with probability $\epsilon'$ or less), party $C$ might output $0$ with any probability.
So in total, the probability that $C$ outputs $0$ is no greater than $\left(1-\epsilon'\right)\cdot\frac{1}{2}+1\cdot\epsilon'$
All nonfaulty parties, including $C$, must output $0$ with probability $\frac{1}{2}+\epsilon$ or more.
Therefore, pick an $\epsilon'$ such that:
\begin{align*}
    \left(1-\epsilon'\right)\cdot\frac{1}{2}+1\cdot\epsilon' & < \frac{1}{2}+\epsilon\\
    \frac{1}{2}-\frac{1}{2}\cdot\epsilon'+\epsilon' & < \frac{1}{2}+\epsilon\\
    \frac{1}{2}+\frac{1}{2}\cdot\epsilon' & < \frac{1}{2}+\epsilon\\
    \epsilon' & < 2\epsilon
\end{align*}
which reaches a contradiction, completing our proof.
A short sketch of how to extend the proof to any $n$ such that $4t\geq n \geq 3t+1$ or to multivalued secrets is provided in the appendix.

\section{Strong Common Coin}

\gs{This comment is true in this part as well: I've said where things will only almost-surely happen in the protocol in a pretty clunky way. What is the right way to do that? Also, as Danny noted, I probably need to be more precise about the difference between completing the protocol, outputting a value, etc.}

The main goal of this section is to construct a strong common coin primitive.
This primitive is defined as follows:
\IA{maybe move this to top of section - at least talk about it? explain why this is new - its has a stronger correctness property}

\begin{definition}
Protocol $CC$ is an $\epsilon$-biased almost-surely terminating common coin protocol if the following properties hold:
\begin{enumerate}
    \item \textbf{Termination.} If all nonfaulty parties participate in the $CC$ protocol they almost-surely complete it. 
    Furthermore, if some nonfaulty party completes protocol $CC$, every nonfaulty party that begins the protocol almost-surely completes it as well.
    \item \textbf{Correctness.} For every value $b\in \left\{0,1\right\}$, there is at least a $\frac{1}{2}-\epsilon$ probability that every nonfaulty party that completes the protocol outputs $b$.
    Regardless, all nonfaulty parties that complete the protocol output the same value with probability $1$.
\end{enumerate}
\end{definition}
This definition has three natural desired properties of a common coin protocol: 
the protocol almost-surely terminates, it has an arbitrarily small bias (as a parameter of the protocol), and the output value is always agreed upon by all parties.
Previous works have achieved some subset of those properties, but not all three together.
For example, the protocol in \cite{CR93} doesn't always terminate and the parties don't always agree on the output value.
On the other hand, the protocol described in \cite{SVSS} always terminates, but can completely fail $O\left(n^2\right)$ times.
This also means that if just one common coin instance is required, there is no guarantee that protocol will yield the desired properties.

Throughout the following sections assume the number of nonfaulty parties is $t$ such that $3t+1\leq n$.
The following protocols use the protocols $SVSS$ and $BA$, which are resilient to this number of faulty parties.
The SVSS protocol, as defined in \cite{SVSS}, has a designated dealer with some input $s$ and it consists of two sub-protocols, $SVSS-Share$ and $SVSS-Rec$.
\begin{definition} An SVSS protocol has the following properties:
\begin{enumerate}
    \item \textbf{Validity of termination.} If a nonfaulty dealer initiates $SVSS-Share$ and all nonfaulty parties participate in the protocol, then every nonfaulty party eventually completes $SVSS-Share$.
    \item \textbf{Termination.} If a nonfaulty party completes either protocol $SVSS-Share$ or $SVSS-Rec$, then all nonfaulty parties that participate in the protocol eventually complete it. 
    \gs{This is slightly different from the claims in the original SVSS paper (the part about reconstruction isn't there). However, it is true because all messages in the rec protocol are broadcasted, and this change is necessary for the coin's termination property. Do we need to prove that the change is correct?} 
    Moreover, if all nonfaulty parties begin protocol $SVSS-Rec$, then all nonfaulty parties eventually complete protocol $SVSS-Rec$.
    \item \textbf{Binding.} Once the first nonfaulty party completes an invocation of $SVSS-Share$ with session id $\left(c,d\right)$, there is a value $r$ such that either:
    \begin{itemize}
        \item the output of each nonfaulty party that completes protocol $SVSS-Rec$ is $r$; or
        \item there exists a nonfaulty party $P_i$ and a faulty party $P_j$ such that $P_j$ is shunned by $P_i$ starting in session $\left(c,d\right)$.
    \end{itemize}
    \item \textbf{Validity.} If the dealer is nonfaulty with input $s$, then the binding property holds with $r=s$.
    \item \textbf{Hiding.} If the dealer is nonfautly and no nonfaulty party invokes protocol $SVSS-Rec$, then the faulty parties learn nothing about the dealer's value.
\end{enumerate}
\end{definition}

Party $P_i$ shuns party $P_j$ if it accepted messages from it in the current invocation, but won't accept any messages from it in future interactions.
For our purposes it is enough to note that fewer than $n^2$ shunning events can take place overall.
\gs{I don't want to go too much into shunning. Is this enough of an explanation?}

\begin{definition}
An almost-surely terminating binary Asynchronous Byzantine Agreement is a protocol in which each nonfaulty party has an input from $\left\{0,1\right\}$, and the following properties hold:
\begin{enumerate}
    \item \textbf{Termination.} If all nonfaulty parties participate in the protocol, all nonfaulty parties almost-surely eventually complete the protocol.
    Furthermore, if some nonfaulty party completes the protocol, all nonfaulty parties that participate in it do so as well.
    \item \textbf{Validity.} If all nonfaulty parties have the same input $\sigma\in\left\{0,1\right\}$, every nonfaulty party that completes the protocol outputs $\sigma$. 
    \item \textbf{Correctness.} All nonfaulty parties that complete the protocol output the same value $\sigma\in\left\{0,1\right\}$.
\end{enumerate}
\end{definition}

Let $SVSS$ be a protocol with the $SVSS$ properties, and $BA$ be an almost-surely terminating binary Asynchronous Byzantine Agreement protocol, as described in \cite{SVSS}.
Both of these protocols are resilient to $t$ faulty processors such that $3t+1\leq n$.

In addition to these two protocol, the common coin protocol requires a protocol for agreeing on a common subset of parties for which some condition holds.
In order to do that, in the protocol each party $P_i$ employs a "dynamic predicate" $Q_{ir}$ for each round $r$.
Intuitively $Q_{ir}\left(j\right)$ denotes whether $P_i$ saw that some irreversible condition holds with regard to $P_j$. 
For every value $j\in\left[n\right]$, $Q_{ir}\left(j\right)\in\left\{0,1\right\}$ at any given point in time.
Initially, $\forall j\in\left[n\right]$ $Q_{ir}\left(j\right)=0$, and for any such $j$, $Q_{ir}\left(j\right)$ can turn into 1, but not back to 0.
The idea of a dynamic predicate and for the protocol below are described in \cite{BKR94}.

\begin{definition}
Protocol $CS$ is a common subset protocol, with a dynamic predicate $Q_i$ and a number $k\leq n$ as input, if it has the following properties:
\begin{enumerate}
    \item \textbf{Termination.} If all nonfaulty parties invoke the protocol, and there exists a set $I\subseteq\left[n\right]$ such that:
    \begin{itemize}
        \item $\left|I\right|\geq k$, and
        \item for every nonfaulty party $P_i$, eventually $\forall j\in I\ Q_i\left(j\right)=1$,
    \end{itemize}
    then all nonfaulty parties almost-surely complete the invocation of $CS$. 
    
    Furthermore, if some nonfaulty party completes protocol $CS$ and if for every pair of nonfaulty parties $P_i,P_j$ that participate in the $CS$ protocol and value $k\in\left[n\right]$ if $Q_{i}\left(k\right)=1$ then eventually $Q_{j}\left(k\right)=1$, then every nonfaulty party that participates in the protocol almost-surely completes it as well. 
    
    \gs{We should be able to strengthen this to "will terminate" by amending BA's properties to something like "once the first nonfaulty party completes the BA protocol, all nonfaulty parties will complete it in constant time" (which is also true).}
    
    \item \textbf{Correctness.} All nonfaulty parties that complete an invocation of $CS$ output the same set $S\subseteq \left[n\right]$. 
    Furthermore, $\left|S\right|\geq k$ and for every $j\in S$ there exists a nonfaulty party $P_i$ such that $Q_i\left(j\right)=1$.
\end{enumerate}
\end{definition}

A construction of a common subset protocol resilient to $t$ faulty parties such that $3t+1\leq n$ is shown and proven with slight changes in \cite{BKR94}.
For completeness, another construction and proof with the aforementioned properties is shown in the appendix.
From this point on, assume the $CommonSubset$ protocol is a common subset protocol.

Using the previously discussed primitives, the rest of this section describes and proves the correctness of a common coin protocol.
Intuitively, in the protocol several weak coins are flipped using the $SVSS$ protocol.
These coins should behave as fully unbiased coin in most cases, but $n^2$ of these coins could fail because the $SVSS$ protocol could fail $n^2$ times.
In this context a coin failing means that it can be totally biased, or not agreed upon.
This means that enough weak coins need to be flipped so that the $n^2$ failures are not significant.
The number of weak coin flips is set to be proportionate to $n^4$ and to a function of the acceptable bias in the final coin, and the output is the value output in the majority of the rounds. 
From the properties of the binomial distribution the $n^2$ faulty coin flips should not significantly bias the result given that around $n^4$ coins are flipped.
\gs{I added this intuitive explanation which should also provide insight into what the proof will look like.
Is this helpful at all?}

\begin{algorithm}[H]
\caption{$CoinFlip\left(\epsilon\right)$}
\underline{Code for $P_i$}:
\begin{itemize}
\item Let $k=4\ceil{\left(\frac{e}{\epsilon\cdot\pi}\right)^2n^4}$
\item For $r$ = 1 to k:
\begin{enumerate}
\item Sample $b_{ir}\leftarrow\left\{0,1\right\}$ uniformly. Call $SVSS-Share_{ir}\left(b_{ir}\right)$ as dealer.
\item Participate in $SVSS-Share_{jr}$ with $P_j$ as dealer for every $j\in\left[n\right]$. 

Note this means the party begins participating in iteration $r$'s $SVSS-Share$ invocations only after completing iteration $r-1$.
\item Define the dynamic predicate $Q_{ir}$ as follows for every $j\in\left[n\right]$:

$Q_{ir}\left(j\right)=
\begin{cases}
1 & \text{if $SVSS-Share_{jr}$ has been completed} \\
0 & else
\end{cases}$
\item Continually participate in $CommonSubset_{r}\left(Q_{ir},n-t\right)$, denote its output as $S_{ir}$.

\item After $CommonSubset_{r}$ terminates, invoke $SVSS-Rec_{jr}$ for every $j\in S_{ir}$, let the reconstructed value be $b_{ijr}$.

\item For every $j\in S_{ir}$ compute $b'_{ijr}=b_{ijr}\mod 2$.

Compute $b'_{ir}=\bigoplus_{j\in S_{ir}}b'_{ijr}$ and continue to the next iteration.

\end{enumerate}

\item After completing the final iteration, compute $b'_i=\majority_{r\in \left[k\right]}\left\{b'_{ir}\right\}$.
\item Participate in a final $BA$ invocation with input $b'_i$.
After completing the $BA$ invocation, output its output.
In addition, continue participating in all relevant invocations of $BA$, $SVSS$ and $CommonSubset$ until they terminate.
\end{itemize}
\end{algorithm}
\gs{As Danny noted, the complexity doesn't go down with the number of shunning events and there is no way to parmaterize that here. It might be better to use $k$ as a parameter and try to work $\epsilon$ out as a function of $k$.}

\begin{theorem}
\label{thm:coinflip}
For every $\epsilon\in\left(0,\frac{1}{2}\right)$ and $t$ faulty processors such that $3t+1\leq n$, protocol $CoinFlip\left(\epsilon\right)$ is an $\epsilon$-biased almost-surely terminating common coin protocol.
\end{theorem}

\begin{proof}
Each property is proven individually. Throughout this proof let $k=4\ceil{\left(\frac{e}{\epsilon\cdot \pi}\right)^2n^4}$ as defined in the protocol.

\textbf{Termination}
First show that if all nonfaulty parties participate in the $CoinFlip\left(\epsilon\right)$ protocol they all almost-surely complete it.
In order to do that, we first show that if all nonfaulty parties start the $r$'th iteration of protocol $CoinFlip\left(\epsilon\right)$, then they all almost-surely complete it. 
Note that all nonfaulty parties continue participating in all $SVSS$ and $CommonSubset$ invocations even after completing the $CoinFlip$ protocol, so if all nonfaulty parties started participating in them, their termination properties continue to hold.
If all nonfaulty parties start the $r$'th iteration of $CoinFlip\left(\epsilon\right)$, every nonfaulty party $P_i$ samples a random value $b_{ir}$, invokes $SVSS-Share_{ir}$ as dealer, and participates in $SVSS-Share_{jr}$ with $P_j$ as dealer for every $j\in\left[n\right]$.
From the Termination property of $SVSS$, since all nonfaulty parties participate in $SVSS-Share_{jr}$ for every nonfaulty dealer $P_j$, all nonfaulty parties eventually complete $SVSS-Share_{jr}$.
Once party $P_i$ completes $SVSS-Share_{jr}$, $Q_{ir}\left(j\right)$ becomes 1.
This means that there exists a set $I\subseteq\left[n\right],\left|I\right|\geq n-t$ such that for every nonfaulty party $P_i$, eventually $\forall j\in I \  Q_{ir}\left(j\right)=1$.
In addition, all nonfaulty parties participate in $CommonSubset_r$ because they started iteration $r$ and continue participating in it even after completing $CoinFlip$ until $CommonSubset_r$ terminates locally.
From the Termination property of $CommonSubset$, all nonfaulty parties almost-surely complete $CommonSubset_r$.
From the Correctness property of $CommonSubset$, for every $j\in S_r$, $Q_{ir}\left(j\right)=1$ for at least one nonfaulty party $P_i$.
This means that for every $j\in S_r$ at least one nonfaulty party completed $SVSS-Share_{jr}$.
From the Termination property of $SVSS$, all other nonfaulty parties complete $SVSS-Share_{jr}$ as well.
After that, all nonfaulty parties reach step 5 of the iteration, and invoke $SVSS-Rec_{jr}$ for every $j\in S_r$.
Again, from the Termination property of $SVSS$, all nonfaulty parties complete $SVSS-Rec_{jr}$ for every $j\in S_r$.
After that, all nonfaulty parties perform local computations in step 6, and reach the end of the iteration.

Since all parties start with the same parameter $\epsilon$ they all compute the same value $k$.
Note that this means that all nonfaulty parties begin the first iteration, and won't stop before completing the $k$'th iteration.
Using a simple inductive argument, all nonfaulty parties almost-surely complete $k$ iterations.
After completing all $k$ iteration, every nonfaulty party then performs a local computation and participates in the last $BA$ invocation.
From the Termination property of the $BA$ protocol, all nonfaulty parties almost-surely complete that $BA$ invocation, and then output its value and complete the protocol.

For the second part of the property, assume some nonfaulty party $P_i$ completed the $CoinFlip$ protocol.
Before doing that, it must have completed the $CommonSubset_r$ protocol for every $r\in\left[k\right]$ and output some set $S_r$.
It must have also completed the $SVSS-Rec_{jr}$ protocol for every $r\in\left[k\right],j\in S_r$, and the final $BA$ protocol.
Now observe some other nonfaulty party $P_l$ that participates in the protocol.
For every nonfaulty party $P_k$ and value $j\in\left[n\right]$, if $Q_{kr}\left(j\right)=1$ it must have first completed the invocation of $SVSS-Share_{jr}$.
From the Termination property of $SVSS$, every other nonfaulty party $P_m$ that participates in $SVSS-Share_{jr}$ completes the protocol as well and sets $Q_{mr}\left(j\right)=1$.
Note that every nonfaulty party that participates in the $CommonSubset_{r}$ protocol also participates in each of the relevant $SVSS-Share$ invocations.
Therefore the conditions of the second part of the Termination property of $CommonSubset$ hold, and thus if $P_l$ participates in $CommonSubset_r$ it  almost-surely completes it as well, and from the Correctness property it outputs $S_r$ as well.
$P_l$ then calls $SVSS-Rec_{jr}$ for every $j\in S_r$ and since those are the same invocations that $P_i$ completed, $P_l$ completes them as well.
This mean that for every $r\in\left[k\right]$, $P_l$ almost-surely completes $CommonSubet_r$ and $SVSS-Rec_{jr}$ for every $j\in S_r$, after which it continues to the next iteration.
After completing all $k$ iterations, $P_l$ performs some local computations and participates in the $BA$ protocol as well.
Since $P_i$ completed the $BA$ protocol, $P_l$ must almost-surely complete the protocol as well, and then complete the protocol.

\textbf{Correctness}
Every nonfaulty party that completes the protocol outputs the value it output in the final $BA$ protocol.
From the Correctness property of $BA$, all nonfaulty parties output the same value in the $BA$ protocol, and thus they all output the same value in the $CoinFlip$ protocol.
This proves the second part of the property

We now turn to deal with the first part of the property.
Every nonfaulty party that completes the $CoinFlip$ protocol must have completed all iterations of the loop in protocol $CoinFlip$.
In each iteration, from the correctness property of $CommonSubset$ there exists some set $S_r$ such that every nonfaulty party that completes $CommonSubset_r$ outputs $S_r$.
From the Correctness property of $CommonSusbet$, at the time $CommonSubset_r$ is completed, for every $j\in S_r$ there exists some nonfaulty party $P_i$ such that $Q_{ir}\left(j\right)=1$.
$P_i$ only sets $Q_{ir}\left(j\right)=1$ if it has has already completed $SVSS-Share_{jr}$.
In other words, at the time some nonfaulty party completes $CommonSubset_r$ there exists some nonfaulty party that completes $SVSS-Share_{jr}$ for every $j\in S_r$.
From the binding property of $SVSS$, at that time some value $s'_{jr}$ is set such that every nonfaulty party that completes $SVSS-Rec_{jr}$ either outputs $s'_{jr}$, or some nonfaulty party shuns some faulty party starting in that $SVSS$ session. 
Denote $c'_{jr}=s'_{jr}\mod 2$, and $c'_r=\bigoplus_{j\in S_r} c'_{jr}$.
Note that $s'_{jr}$ is supposed to be either $0$ or $1$ but in the case of sharing over a large field, this cannot be enforced for faulty dealers.

For every $j\in S_r$, no nonfaulty party invokes $SVSS-Rec_{jr}$ before completing the  $CommonSubset_r$ invocation, at which time $S_r$ is set already.
From the hiding property of $SVSS$, before some nonfaulty party invokes $SVSS-Rec_{jr}$ for any nonfaulty dealer $P_j$, the faulty (and nonfaulty) parties' view is distributed independently of the value $b_{jr}$ shared by $P_j$.
This also means that the values shared by any nonfaulty party $P_j$ such that $j\in S_r$ are entirely independent of other values shared by all other parties in $S_r$.
From the validity property of $SVSS$, $c'_{jr}=b_{jr}$ for every nonfaulty $P_j$.
Since $\left|S_r\right|\geq n-t$, there exists at least one nonfaulty party $P_j$ such that $j\in S_r$.
Note that $c'_{r}=0$ if and only if $\bigoplus_{l\in S_r\setminus \left\{j\right\}} c'_{lr}=c'_{jr}=b_{jr}$. $b_{jr}$ is sampled uniformly from $\left\{0,1\right\}$ and entirely independently from the rest of the values, and thus the probability that $c'_r=0$ for any $r\in\left[k\right]$ is exactly $\frac{1}{2}$.
Using similar arguments it can also be shown that the values $c'_{r}$ are independent of values computed in all other iterations.

For each $r\in\left[k\right]$ either every nonfaulty party $P_i$ that completes the $r$'th iteration computes $b'_{ir}=c'_r$ or some nonfaulty party shuns some faulty party starting in iteration $r$.
Overall, there can occur fewer than $n^2$ shunning events, and thus for at least $k-n^2$ different iterations every nonfaulty party $P_i$ that completes the $r$'th iteration computes $b'_{ir}=c'_r$.
This means that if $\left|\left\{r|c'_r=1\right\}\right|>\frac {k}{2}+n^2$, then regardless of the faulty parties' actions, every nonfaulty party $P_i$ that completes all $k$ iterations outputs $b'_{ir}=c'_r=1$ for at least $\left\lfloor{\frac{k}{2}}\right\rfloor + 1$ of those iteration, and thus inputs 1 to the $BA$ invocation at the end of the protocol.
From the correctness property of $BA$, if every nonfaulty party that participates in a $BA$ invocation inputs the value $1$, then every nonfaulty party that completes the invocation outputs $1$.
In that case, all nonfaulty parties output $1$ in the end of the $CoinFlip$ protocol.
The exact same argument can be made stating that all nonfaulty parties output $0$.
The proof that these events take place with probability $\frac{1}{2}-\epsilon$ at the very least follows from well-known properties of the binomial distribution and is therefore moved to the appendix.
\end{proof}

\section{Fair Agreement}

\gs{First, I've said where things will only almost-surely happen in the protocol in a pretty clunky way. What is the right way to do that? Also, as Danny noted, I probably need to be more precise about the difference between completing the protocol, outputting a value, etc.}

This section deals with constructing a Byzantine Agreement protocol with strong properties.
First of all, the regular notions of Correctness (i.e. agreement) and Termination are preserved.
In addition to that, a stronger notion of Validity is achieved in the case of multivalued agreement. 
If all nonfaulty processors have the same input $\sigma$, they all output $\sigma$;
however, if that is not the case, the probability that all nonfaulty parties output some nonfaulty party's input is at least $\frac{1}{2}$.
This also nicely extends to natural notions of fairness in the case of a non-Byzantine adversary.

\begin{definition}
A Fair Byzantine Agreement protocol has the following properties:
\begin{enumerate}
    \item \textbf{Termination.} If all nonfaulty parties participate in the protocol, they almost-surely complete it.
    Furthermore, if some nonfaulty party completes the protocol, all other nonfaulty parties that participate in it almost-surely complete it as well.
    \item \textbf{Validity.} If all nonfaulty parties have the same input to the protocol, they output that value.
    Otherwise, with probability at least $\frac{1}{2}$, all nonfaulty parties output some nonfaulty party's input.
    \item \textbf{Correctness.} All nonfaulty parties that complete the protocol output the same value.
\end{enumerate}
\end{definition}

The goal in this section is to design a Fair Byzantine Agreement protocol.
In order to do so, a protocol for choosing one element out of $m$ elements in an almost fair way is described.

\begin{definition}
A Fair Choice protocol has the following properties if all nonfaulty parties that participate in it have the same input $m\geq 3$:
\begin{enumerate}
    \item \textbf{Termination.} If all nonfaulty parties participate in the protocol they all almost-surely complete it.
    Furthermore, if some nonfaulty party completes the protocol, all other nonfaulty parties that participate in it almost-surely complete it as well.
    \item \textbf{Validity.} For any set $G\subseteq\left\{0,\ldots,m-1\right\}$ such that $\left|G\right|> \frac{m}{2}$ the probability that all nonfaulty parties that complete the protocol output some $i\in G$ is at least $\frac{1}{2}$.
    \item \textbf{Correctness.} All nonfaulty parties that complete the protocol output the same value $i\in\left\{0,\ldots,m-1\right\}$.
\end{enumerate}
\end{definition}

\gs{What is the correct way to deal with the input $m$? Do I just say that this is a part of the protocol or do I deal with it as an input and condition upon having the same input?}

\begin{algorithm}[H]
\caption{$FairChoice(m)$}

\underline{Code for $P_i$}:
\begin{enumerate}
\item Set $N=2^l$ for the smallest $l\in\mathbb{N}$ such that $4m^2\geq N\geq 2m^2$ and set $\epsilon=\frac{1}{100m\log_2 m}$.
\item For every $i\in\left[l\right]$ participate in $CoinFlip_i\left(\epsilon\right)$ and let the $i$'th output be $b_i$.
\item Let $r$ be the number whose binary representation is $b_1b_2\ldots b_l$.
Output $r \mod m$.
\end{enumerate}
\end{algorithm}

\begin{restatable}{theorem}{fairchoice}
\label{thm:fairchoice}
FairChoice is a Fair Choice protocol for any number of faulty parties $t$ such that $3t+1\leq n$.
\end{restatable}
The proof is provided in the appendix.
A Fair Byzantine Agreement protocol that uses the Fair Choice protocol is described below.
In this Fair Byzanting Agreement protocol, each party $P_i$ has some input $x_i$.
The construction makes use of a Broadcast protocol.

\begin{definition}
A Broadcast protocol is a protocol with a designated sender $P_i$ with some input $v$, which has the following properties:
\begin{enumerate}
    \item \textbf{Termination.} If $P_i$ is nonfaulty and all nonfaulty parties participate in the protocol, they all complete the protocol.
    Furthermore, if some nonfaulty party completes the protocol, every other nonfautly party that participates in it does so as well.
    \item \textbf{Validity.} If $P_i$ is nonfaulty, every nonfaulty party that completes the protocol outputs $v$.
    \item \textbf{Correctness.} All nonfaulty parties that complete the protocol output the same value.
\end{enumerate}
\end{definition}
\noindent Let A-Cast be a Broadcast protocol, for example as described in \cite{B87}.

\begin{algorithm}[H]
\caption{$FBA$}

\underline{Code for $P_i$ with input $x_i$}:
\begin{enumerate}
\item A-Cast $x_i$ and participate in every other party's $A-Cast$.
Denote the output of $P_j$'s A-Cast to be $x'_j$.
\item Define the dynamic predicate $Q_i$ as follows:

$Q_{i}\left(j\right)=
\begin{cases}
1 & \text{if $P_j$'s A-Cast has been completed} \\
0 & else
\end{cases}$
\item Continually participate in $CommonSubset\left(Q_i,n-t\right)$.
\item After completing the $CommonSubset$ protocol, let $S$ be the protocol's output and let $m=\left|S\right|$.
Wait to complete $P_j$'s A-Cast for every $j\in S$.
\item If there exists some value $x$ such that $\left|\left\{x'_j=x|j\in S\right\}\right|> \frac{m}{2}$, output $x$ and complete the protocol.
Otherwise, continue to the next step.
\item Participate in $FairChoice\left(m\right)$, and let the output be $k$.
\item Let $j$ be the $k$'th biggest value in $S$, with $0$ being understood as the biggest value, $1$ as the second biggest, etc.
\item Output $x'_j$.
\end{enumerate}
\end{algorithm}

\begin{restatable}{theorem}{fba}
\label{thm:fba}
Protocol $FBA$ is a Fair Byzantine Agreement protocol for any number of faulty parties $t$ such that $3t+1\leq n$.
\end{restatable}

Intuitively, each party A-Casts its input value, and the parties agree on a subset of parties of size $n-t$ at the very least whose values have been received using the $CommonSubset$ protocol.
If all nonfaulty parties have the same input, they will see that a majority of the parties sent the same value and output that value in line 5, achieving the first part of the Validity property.
Otherwise, the parties choose the value sent by one of those parties "almost fairly" using the $FairChoice$ Protocol.
Since more than half of the parties in the agreed upon subset are nonfaulty, the probability that a nonfaulty party will be chosen is at least $\frac{1}{2}$.
A formal proof is provided in the appendix.

\begin{acks}
The authors would like to thank the anonymous referees for
their valuable comments and helpful suggestions.
This work was supported by the \grantsponsor{RE133010}{HUJI Federnann Cyber Security Research Center}{https://www.gov.il/he/Departments/General/academicresearchcenter} in conjunction with the Israel National Cyber Directorate (INCD) in the Prime Minister's Office under Grant No.: \grantnum{RE133010}{3011004045}.
\end{acks}

\bibliographystyle{plain}
\bibliography{bibfile}

\newpage
\appendix
\section{Proofs of Technical Lemmas in Section \ref{sec:main}}
\indmarginallong*
\begin{proof}
Each of these equalities is shown individually.

\begin{align*}
& \Pr\left[m^0_{AD}=m'_{AD}|\overline{long}\right] =\\
&=\sum_{\substack{m'_{AB}\in M^0_{AB}|\overline{long},\\r'_A\in R^0_A|\overline{long}}} \Pr\left[m^0_{AD}=m'_{AD},m^0_{AB}=m'_{AB},r_A=r'_A|\overline{long}\right] \\
& = \sum_{\substack{m'_{AB}\in M^1_{AB}|\overline{long},\\r'_A\in R^1_A|\overline{long}}} \Pr\left[m^1_{AD}=m'_{AD},m^1_{AB}=m'_{AB},r_A=r'_A|\overline{long}\right] \\
& = \Pr\left[m^1_{AD}=m'_{AD}|\overline{long}\right]
\end{align*}

Note that summing over $m'_{AB}\in M^0_{AB}|\overline{long}$ is the same as summing over $m'_{AB}\in M^1_{AB}|\overline{long}$ because for every $m'_{AB}\in M^0_{AB}|\overline{long}$ there must exist $m'_{AD},r'_A$ such that $\Pr[m_{AD}^{0} = m'_{AD} , m_{AB}^{0} = m'_{AB} ,\newline r^0_A = r' | \overline{long}] \neq 0$. 
From previous observations, this means that  $\Pr[m_{AD}^{1}=m'_{AD},m_{AB}^{1}=m'_{AB},r^1_A=r'_A|\overline{long}]\neq0$ and thus $m'_{AB}\in M^0_{AB}|\overline{long}$ as well.
The same reasoning holds about $r'_A$.
The argument also clearly works in reverse.
This argument can also be made for any other subset of the three variables.

For the second property, note that $\Pr\left[m^0_{AB}=m'_{AB}|\overline{long}\right]\neq 0$ and thus the probability is well defined. In that case:
\begin{align*}
&\Pr\left[m^0_{AD}=m'_{AD}|m^0_{AB}=m'_{AB},\overline{long}\right] =\\
& = \frac{\Pr\left[m^0_{AD}=m'_{AD},m^0_{AB}=m'_{AB}|\overline{long}\right]}
{\Pr\left[m^0_{AB}=m'_{AB}|\overline{long}\right]} \\
& = \frac{\Pr\left[m^1_{AD}=m'_{AD},m^1_{AB}=m'_{AB}|\overline{long}\right]}
{\Pr\left[m^1_{AB}=m'_{AB}|\overline{long}\right]} \\
& = \Pr\left[m^1_{AD}=m'_{AD}|m^1_{AB}=m'_{AB},\overline{long}\right]
\end{align*}
It is important to notice that all of those arguments could have been made with any subset of the three random variables described in the lemma.
\end{proof}

\samplable*
\begin{proof}
We go through each sampled value and check if the distribution is well-defined.
First, $D$ samples $s_A\leftarrow R^0_A|\overline{long}$.
From the definitions of $\epsilon'$ and the corresponding $N$, the probability that $A$'s view throughout protocol $S$ is of length greater than $N$ is no greater than $\epsilon'$.
This means that the event $\overline{long}$ happens with probability $1-\epsilon' > 0$ at the very least, and thus there must also exist some value $r_A\in R^0_A|\overline{long}$.
$D$ then samples $s_{AB}\leftarrow M^0_{AB}|r^0_A=s_A,\overline{long}$. Since $\Pr\left[r^0_A=r_A|\overline{long}\right]\neq 0$, there must be some set of messages $m'_{AB}\in M^0_{AB}|\overline{long}$ such that $\Pr\left[m^0_{AB}=m'_{AB},r^0_A=s_A|\overline{long}\right]\neq 0$ and thus the distribution is well defined.
The argument for $s_{AD}$ is identical.
$D$ then samples $s_B\leftarrow R^1_B|m^1_{AB}=s_{AB},\overline{long}$.
Following similar arguments, $\Pr\left[m^0_{AB}=s_{AB}|\overline{long}\right]\neq 0$ and thus from corollary \ref{col:indmarginallong}, $\Pr\left[m^1_{AB}=s_{AB}|\overline{long}\right]\neq 0$.
Now, following similar arguments both $s_B$ and $s_{BD}$ are sampled from well-defined distributions.
$D$ can easily sample from these distributions by simulating all runs with parties $A,B$ and $D$ that take no more than $N$ rounds to terminate.
This is possible because of the assumption that the size of messages and randomness in each round is bounded.
If that is not the case, $D$ can simulate the protocol step by step and sample values that way.
\end{proof}

\distribution*
\begin{proof}
The random variable $v_{A}^{0}$ is defined to be party $A$'s view during protocol $S$ with a nonfaulty dealer $D$ sharing the value $s=0$, and a faulty $C$ which remains silent.
Since no messages are received from party $C$, $A$'s view consists of $m^0_{AB},m^0_{AD},r^0_A$ \gs{Technically, $A$ could also send $C$ some messages, but since no response is received, those messages are deterministic given the rest of the view}.
In the run described in the lemma no messages are sent or received from party $C$ either and thus party $A$'s
view consists of $s_{AB},s_{AD},r_{A}$.
Technically the ordering could also matter, but note that the scheduling is deterministic and looks identical in both runs, so the order in which messages are received is ignored.

Observe some $m'_{AB},m'_{AD},r'_{A}$ such that $\Pr[m^0_{AB}=m'_{AB},m^0_{AD}=m'_{AD},r_A=r'_A|\overline{long}]\neq 0$:

\begin{align*}
&\Pr\left[m_{AB}^{0}=m'_{AB},m_{AD}^{0}=m'_{AD},r^0_{A}=r'_{A}|\overline{long}\right] =\\
& =\Pr\left[m_{AD}^{0}=m'_{AD}|m_{AB}^{0}=m'_{AB},r^0_{A}=r'_{A},\overline{long}\right] \\
& \cdot \Pr\left[m_{AB}^{0}=m'_{AB}|r^0_{A}=r'_{A},\overline{long}\right]\cdot\Pr\left[r^0_{A}=r'_{A}|\overline{long}\right]
\end{align*}

On the other hand:

\begin{align*}
&\Pr\left[s_{AB}=m'_{AB},s_{AD}=m'_{AD},r_{A}=r'_{A}|G\right] =\\ 
&=\Pr\left[s_{AD}=m'_{AD}|s_{AB}=m'_{AB},r_{A}=r'_{A},G\right]\\
& \cdot \Pr\left[s_{AB}=m'_{AB}|r_{A}=r'_{A},G\right]\Pr\left[r_{A}=r'_{A}|G\right] \\
& =\Pr\left[s_{AD}=m'_{AD}|s_{AB}=m'_{AB},s_{A}=r'_{A}\right]\\
& \cdot \Pr\left[s_{AB}=m'_{AB}|s_{A}=r'_{A}\right]\Pr\left[s_{A}=r'_{A}\right] \\
& =\Pr\left[m^0_{AD}=m'_{AD}|m^0_{AB}=m'_{AB},r^0_{A}=r'_{A},\overline{long}\right]\\
& \cdot \Pr\left[m^0_{AB}=m'_{AB}|r^0_{A}=r'_{A},\overline{long}\right]\Pr\left[r^0_{A}=r'_{A}|\overline{long}\right] \\
& = \Pr\left[m^0_{AB}=m'_{AB},m^0_{AD}=m'_{AD},r^0_A=r'_A|\overline{long}\right]
\end{align*}

Where the second to last equality stems from the definitions of the random variables $s_A,s_{AB},s_{AD}$.

The analysis for $B$'s view can be done in a similar fashion, finding that:

\begin{align*}
&\Pr\left[s_{AB}=m'_{AB},s_{BD}=m'_{BD},r_{B}=r'_{B}|G\right] =\\
& =\Pr\left[s_{BD}=m'_{BD}|s_{AB}=m'_{AB},r_{B}=r'_{B},G\right]\\
& \cdot \Pr\left[r_{B}=r'_{B}|s_{AB}=m'_{AB},G\right]\Pr\left[s_{AB}=m'_{AB}|G\right] \\
& =\Pr\left[s_{BD}=m'_{BD}|s_{AB}=m'_{AB},s_{B}=r'_{B}\right]\\
& \cdot \Pr\left[s_{B}=r'_{B}|s_{AB}=m'_{AB}\right]\Pr\left[s_{AB}=m'_{AB}|G\right] \\ 
& =\Pr\left[m^1_{BD}=m'_{BD}|m^1_{AB}=m'_{AB},r^1_{B}=r'_{B},\overline{long}\right]\\
& \cdot \Pr\left[r^1_{B}=r'_{B}|m^1_{AB}=m'_{AB},\overline{long}\right]\Pr\left[s_{AB}=m'_{AB}|G\right]
\end{align*}
Where the final equality stems from the definition of the random variables $s_{BD},s_B$.
Now observe the messages between parties $A$ and $B$:

\begin{align*}
&\Pr\left[s_{AB}=m'_{AB}|G\right] =\\
& = \sum_{r'_A\in R^0_{A}|\overline{long}} \Pr\left[s_{AB}=m'_{AB}|r_{A}=r'_{A},G\right]\Pr\left[r_A=r'_A|G\right] \\
& = \sum_{r'_A\in R^0_{A}|\overline{long}} \Pr\left[s_{AB}=m'_{AB}|s_{A}=r'_{A}\right]\Pr\left[s_A=r'_A\right] \\
& = \sum_{r'_A\in R^0_{A}|\overline{long}} \Pr\left[m^0_{AB}=m'_{AB}|r^0_{A}=r'_{A},\overline{long}\right]\Pr\left[r^0_A=r'_A|\overline{long}\right] \\
& = \Pr\left[m^0_{AB}=m'_{AB}|\overline{long}\right] = \Pr\left[m^1_{AB}=m'_{AB}|\overline{long}\right]
\end{align*}
Where the third equality stems from the definitions of $s_A$ and $s_{AB}$, and the last equality stems from corollary \ref{col:indmarginallong}.
Completing the original analysis:

\begin{align*}
&\Pr\left[s_{AB}=m'_{AB},s_{BD}=m'_{BD},r_{B}=r'_{B}|G\right] =\\
& =\Pr\left[m^1_{BD}=m'_{BD}|m^1_{AB}=m'_{AB},r^1_{B}=r'_{B},\overline{long}\right]\\
& \cdot \Pr\left[r^1_{B}=r'_{B}|m^1_{AB}=m'_{AB},\overline{long}\right]\Pr\left[s_{AB}=m'_{AB}|G\right] \\
& =\Pr\left[m^1_{BD}=m'_{BD}|m^1_{AB}=m'_{AB},r^1_{B}=r'_{B},\overline{long}\right]\\
& \cdot \Pr\left[r^1_{B}=r'_{B}|m^1_{AB}=m'_{AB},\overline{long}\right]\Pr\left[m^1_{AB}=m'_{AB}|\overline{long}\right] \\
& = \Pr\left[m^1_{AB}=m'_{AB},m^1_{BD}=m'_{BD},r^1_B=r'_B|\overline{long}\right]
\end{align*}

Since an equality holds for every nonzero-probability event and both views must define probability spaces, the distributions must be the same.
\end{proof}

\Aattack*
\begin{proof}
The scheduling is identical to the scheduling described in the previous lemma, and party $A$ similarly acts as a nonfaulty party throughout all of protocol $S$. 
Following the exact same arguments, parties $A,B$ and $D$ must complete protocol $S$ without party $C$ sending or receiving any messages.
Similarly define $\hat{m}_{XY}$ to be the messages party $X$ and $Y$ exchanged throughout protocol $S$, and $\hat{r}_X$ to be party $X$'s randomness throughout the protocol.

After completing protocol $S$, party $A$ simulates all runs in which $\overline{long}$ takes place when a nonfaulty dealer shares the value $0$.
If there is no such run in which the messages $\hat{m}_{AB}$ are exchanged between parties $A$ and $B$, party $A$ acts as a nonfaulty processor throughout protocol $R$.
Otherwise, using those simulations, party $A$ samples random values $\hat{s}_A\leftarrow R^0_A|m^0_{AB}=\hat{m}_{AB},\overline{long}$ and messages $\hat{s}_{AD}\leftarrow M^0_{AD}|m^0_{AB}=\hat{m}_{AB},r^0_A=\hat{s}_A,\overline{long}$.
Note that in this case clearly $\Pr\left[m^1_{AB}=m_{AB}|\overline{long}\right]\neq 0$, and thus also $\Pr\left[m^0_{AB}=m_{AB}|\overline{long}\right]\neq 0$ from corollary \ref{col:indmarginallong}.
This means that the above distributions are well-defined.
From this point on, party $A$ acts as a nonfaulty party would act with a view consisting of $\hat{m}_{AB},\hat{s}_{AD},\hat{s}_A$.
The scheduling from this point on is identical to the scheduling described in the previous lemma.

Recall that $m_{XY}$ is defined as the messages exchange by parties $X$ and $Y$ and $r_x$ is defined as $X$'s randomness throughout the attack described in claim \ref{claim:firstattack}. 
Now observe a snapshot of the values party $A$ saw throughout protocol $S$ and the values party $B$ claims it saw throughout the protocol.
For any values $r'_A,r'_B,m'_{AB},m'_{BD},m'_{AD}$ such that
$\Pr[m_{AB}=m'_{AB},m_{AD}=m'_{AD},m_{BD}=m'_{BD},r_A=r'_A,r_B=r'_B|G]\neq 0$
first analyse the variable $s_A$:
\begin{align*}
    &\Pr\left[s_A=r'_A|s_{AB}=m'_{AB},s_{BD}=m'_{BD},s_B=r'_B\right]= \\
    &\Pr\left[s_A=r'_A|s_{AB}=m'_{AB}\right]= \\
    &\frac{\Pr\left[s_{AB}=m'_{AB}|s_A=r'_A\right]
    \Pr\left[s_A=r'_A\right]}
    {\sum_{\bar{r}_A\in R^0_A|\overline{long}} \Pr\left[s_{AB}=m'_{AB}|s_A=\bar{r}_A\right]\Pr\left[s_A=\bar{r}_A\right]}= \\
    &\frac{\Pr\left[m^0_{AB}=m'_{AB}|r^0_A=r'_A,\overline{long}\right]
    \Pr\left[r^0_A=r'_A|\overline{long}\right]}
    {\sum_{\bar{r}_A\in R^0_A|\overline{long}} \Pr\left[m^0_{AB}=m'_{AB}|r^0_A=\bar{r}_A,\overline{long}\right]\Pr\left[r^0_A=\bar{r}_A|\overline{long}\right]}=\\
    &\Pr\left[r^0_A=r'_A|m^0_{AB}=m'_{AB},\overline{long}\right]
\end{align*}
Where the first equality stems from the fact that given $s_A$, $s_{AB}$ is independent of $s_{BD},s_B$, from which the reverse also follows.
In addition, the third equality stems from the definition of $s_{AB}$.
Now continue the analysis in a similar fashion to before:
\begin{align*}
    &\Pr\left[m_{AB}=m'_{AB},m_{AD}=m'_{AD},m_{BD}=m'_{BD},r_A=r'_A,r_B=r'_B|G\right] \\
    &=\Pr\left[s_{AB}=m'_{AB},s_{AD}=m'_{AD},s_{BD}=m'_{BD},r_A=r'_A,r_B=r'_B|G\right] \\
    &=\Pr\left[s_{AB}=m'_{AB},s_{BD}=m'_{BD},r_B=r'_B|G\right] \\
    &\cdot\Pr\left[s_A=r'_A|s_{AB}=m'_{AB},s_{BD}=m'_{BD},s_B=r'_B\right]  \\
    &\cdot\Pr\left[s_{AD}=m'_{AD}|s_{AB}=m'_{AB},s_{BD}=m'_{BD},s_A=r'_A,s_B=r'_B\right] \\
    &=\Pr\left[m^1_{AB}=m'_{AB},m^1_{BD}=m'_{BD},r^1_B=r'_B|\overline{long}\right] \\
    &\cdot\Pr\left[r^0_A=r'_A|m^0_{AB}=m'_{AB},\overline{long}\right] \\
    &\cdot\Pr\left[m^0_{AD}=m'_{AD}|m^0_{AB}=m'_{AB},r_A=r_A,\overline{long}\right]
\end{align*}
Where the last equality stems from several facts.
From lemma~\ref{lem:distribution}, $\Pr[s_{AB}=m'_{AB},s_{BD}=m'_{BD},r_B=r'_B|G]=\Pr [ m^1_{AB} = m'_{AB} , m^1_{BD} = m'_{BD} , r^1_B = r'_B | \overline{long}]$.
From the definition of the random variable $s_{AD}$, given $s_{AB}$ and $s_A$, the variable $s_{AD}$ is independent of the variables $s_{BD},s_B$, 
and then the equality stems from the definition of $s_{AD}$ and from the previous analysis.

On the other hand, note that since parties $A,B$ and $D$ are acting as nonfaulty parties throughout $S$, their actions are distributed identically to the setting in which $C$ is faulty and silent.
In this setting, the event $\overline{long}$ takes place with probability $1-\epsilon'$ at the very least.
Note that if this event takes place, then processor $A$ sees that the messages $\hat{m}_{AB}$ can be exchanged in some run in which the event $\overline{long}$ takes place, and thus sample some values.
Therefore, conditioned upon the event $\overline{long}$:
\begin{align*}
    &\Pr\left[\hat{m}_{AB}{=}m'_{AB},\hat{s}_{AD}=m'_{AD},\hat{m}_{BD}=m'_{BD},\hat{s}_A=r'_A,\hat{r}_B=r'_B|\overline{long}\right] \\
    &=\Pr\left[\hat{m}_{AB}=m'_{AB},\hat{m}_{BD}=m'_{BD},\hat{r}_B=r'_B|\overline{long}\right] \\
    &{\cdot}\Pr\left[\hat{s}_{A}=r'_A|\hat{m}_{AB}=m'_{AB},\hat{m}_{BD}=m'_{BD},\hat{r}_B=r'_B,\overline{long}\right]  \\
    &{\cdot}\Pr\left[\hat{s}_{AD}{=}m'_{AD}|\hat{m}_{AB}{=}m'_{AB},\hat{m}_{BD}=m'_{BD},\hat{s}_A=r'_A,\hat{r}_B=r'_B,\overline{long}\right] \\
    &=\Pr\left[m^1_{AB}=m'_{AB},m^1_{BD}=m'_{BD},r^1_B=r'_B|\overline{long}\right] \\
    &{\cdot}\Pr\left[r^0_A=r'_A|m^0_{AB}=m'_{AB},\overline{long}\right]  \\
    &{\cdot}\Pr\left[m^0_{AD}=m'_{AD}|m^0_{AB}=m'_{AB},r^0_A=r'_A,\overline{long}\right]
\end{align*}
Where the last equality stems from similar arguments.
First of all note that from $B$'s point of view, party $C$ is acting like a faulty party which is staying silent throughout protocol $S$ 
and parties $A,D$ are acting as nonfaulty parties with $D$ sharing the value $1$.
Therefore, $\Pr[\hat{m}_{AB}=m'_{AB},\hat{m}_{BD}=m'_{BD},\hat{r}_B=r'_B|\overline{long}]=\Pr[m^1_{AB}=m'_{AB},m^1_{BD}=m'_{BD},r^1_B=r'_B|\overline{long}]$.
This also means that if event $\overline{long}$ occurs, party $B$'s view is distributed according to $V^1_B|\overline{long}$.
From the way $\hat{s}_A$ is sampled, given $\hat{s}_{AB}$, the random variable $\hat{s}_A$ is entirely independent of $\hat{m}_{BD},\hat{r}_B$.
Taking that fact into consideration, and looking at the definition of $\hat{s}_A$, $\Pr[\hat{s}_A=r'_A|\hat{m}_{AB}=m'_{AB},\hat{m}_{BD}=m'_{BD},\hat{r}_B=r'_B,\overline{long}]=\Pr[r^0_A=r'_A|m^0_{AB}=m'_{AB},\overline{long}]$.
A similar argument can be made for $\hat{s}_{AD}$.

From this point on the rest of the argument is identical to the argument in the previous lemma, 
finding that if event $\overline{long}$ occurs, the probability that party $B$ outputs $0$ is $\frac{1}{2}$ or less.
Since event $\overline{long}$ occurs with probability $1-\epsilon'$ at the very least, this completes the proof.
\end{proof}

\section{Extending the Impossibility Result}
In order to extend the proof to a multivalued secret, it is enough to note that any protocol in which the dealer can share values from some set $V$ can be used for sharing binary values.
For example, this can be done by mapping the possible values to the values $0$ and $1$ in some predetermined fashion.
Extending the result to any $n$ such that $4t\geq n\geq 3t+1$ requires a more intricate simulation.
If there exists a terminating $\left(\frac{1}{2}+\epsilon\right)$-correct $t$-resilient Byzantine AVSS protocol for some $4t\geq n\geq 3t+1$, then there must also exist such a protocol for $n=4,t=1$.
A sketch for this reduction follows:
\begin{itemize}
    \item Parties $A,B,C$ each simulate $t$ parties running the protocol for the case that $4t\geq n\geq 3t+1$, and party $D$ simulates $n-3t$ parties.
    The dealer must be one of the parties $D$ simulates.
    \item Every time some party needs to send a message between two parties it is simulating, the simulating party just "delivers" the message.
    \item If some message is sent between simulated parties controlled by different parties, the message is sent between the simulating parties, including the ids of the sending and receiving parties.
    When the message is received, the relevant simulating party "delivers" the message to the correct party and continues the simulation accordingly.
    \item Finally, when some simulating party sees that all of the parties it controls completed protocol $R$ and output some value, the simulating party outputs the value which was output by most of the parties it simulated.
\end{itemize}

This is a standard technique.
Note that the message scheduling in the simulation could also take place in the case that there actually are $n$ parties.
In addition, the adversary can only control up to $t$ parties.
In order to give the adversary full control over which parties it controls, this argument can be made with each possible allocation of simulated parties.
Clearly, in every case in which all nonfaulty parties complete protocol $R$ and all output some value, every nonfaulty party will also output the same value, from which all of the properties follow.

\gs{This is very bare-bones and ugly. Is this the level of depth you wanted to go into?}

\section{Construction and Proof of a Common Subset Protocol}
\begin{algorithm}[H]
\caption{$CommonSubset_{r}\left(Q_{ir},k\right)$}

\underline{Code for $P_i$}:
\begin{enumerate}
\item Initialize $c_{ir}=0$.
\item For every $j\in \left[n\right]$, once $Q_{ir}\left(j\right)$ becomes 1, if $c_{ir}<k$, begin participating in $BA_{jr}$ with input 1.
\item If at any point $BA_{jr}$ terminates with output 1 for any $j\in\left[n\right]$, set $c_{ir}=c_{ir}+1$.
\item Once $c_{ir}\geq k$, begin participating in $BA_{jr}$ with input 0 for every $j\in\left[n\right]$ such that $Q_{ir}\left(j\right)=0$ at this point in time.
\item Denote $b_{jr}$ to be the output of $BA_{jr}$.
Output $\left\{j|b_{jr}=1\right\}$.
\item Continue participating in $BA_{jr}$ for every $j\in\left[n\right]$ until they terminate even after completing this invocation of $CommonSubset_r$. \gs{For example, this is weirdly worded and not explicitly defined. What is the right way to deal with this?}
\end{enumerate}
\end{algorithm}

Note that throughout this discussion we assume $k\leq n$ and $3t+1\leq n$.

\begin{lemma}\label{lem:cs_size}
If there exists a set $I\subseteq\left[n\right],\left|I\right|\geq k$ such that for every nonfaulty party $P_i$, eventually $\forall j\in I \  Q_{ir}\left(j\right)=1$ and all nonfaulty parties invoke $CommonSubset_{r}$ for a given r, then at least $k$ invocations of $BA_{jr}$ almost-surely terminate with output 1.
\end{lemma}
\begin{proof}
Note that $c_{ir}$ is incremented only when $BA_{jr}$ terminates with output 1.
In addition, every nonfaulty party $P_i$ inputs 0 to any $BA_{jr}$ invocation only after having $c_{ir}\geq k$.
This means that if some nonfaulty party inputs 0 to some invocation of $BA_{jr}$ then it must have completed at least $k$ prior invocations with output 1.
From the correctness property of protocol $BA$, every other nonfaulty party also outputs 1 for the same invocations of $BA$, which proves our lemma. 
Thus, assume no nonfaulty party inputs the value $0$ to any invocation of $BA_{jr}$ ever for any $j\in\left[n\right]$.
In that case, every nonfaulty party $P_i$ invokes $CommonSubset_r$, and eventually for every $j\in I$ $Q_{ir}\left(j\right)=1$.
Thus every nonfaulty party begins participating in $BA_{jr}$ with input 1 for every $j\in I$.
From the Validity and Termination properties of $BA$ all nonfaulty parties almost-surely complete those invocations of $BA_{jr}$ with output 1.
Since $\left|I\right|\geq k$, this completes the proof.
\end{proof}

\begin{theorem}
Protocol $CommonSubset$ is a common subset protocol for any number of faulty parties $t$ such that $3t+1\leq n$.
\end{theorem}
\begin{proof}
Each property is proven separately.

\textbf{Correctness.}
If two nonfaulty parties $P_i,P_l$ complete $Common\-Subset_r$ then they must have completed $BA_{jr}$ for every $j\in\left[n\right]$.
From the correctness property of $BA$, they completed each of those invocation with the same output $b_{jr}$ and thus both output $S_r=\left\{j|b_{jr}=1\right\}$.
Next, show that for every $j\in S_r$, $Q_{ir}\left(j\right)=1$ for at least one nonfaulty party $P_i$.
Assume by way of contradiction $Q_{ir}\left(j\right) = 0$ for every nonfaulty party $P_i$ for some $j\in S_r$.
If that is the case, and some nonfaulty party completed $CommonSubset_r$, every nonfaulty party that participated in $BA_{jr}$ at that point must have input 0. 
From those parties' point of view, this run is identical to one in which all nonfaulty parties' inputs are 0, and some might be slow.
From the Validity property of $BA$, all nonfaulty parties must have then output 0 in $BA_{jr}$. 
However, in that case $b_{jr}\neq 1$, and thus $j\notin S_r$ reaching a contradiction.
Finally, show that $\left|S_r\right|\geq k$.
Assume by way of contradiction $\left|S_r\right|<k$. In that case, all parties completed all invocations of $BA_{jr}$, with at most $k-1$ terminating with output 1. 
Since nonfaulty parties increment $c_{ir}$ exactly once for every $BA$ session that outputs the value $1$, this means that for every nonfaulty party $P_i$, $c_{ir}<k$. 
Since $k\leq n$, $BA_{jr}$ terminated with output 0 for at least one $j\in\left[n\right]$. 
Observe $BA_{jr}$ for that $j$.
Nonfaulty parties participate in any $BA_{jr}$ session only if either $Q_{ir}\left(j\right) = 1$ or $c_i\geq k$.
Since $c_i<k$, $Q_{ir}\left(j\right)$ must equal $1$ at the time of invoking $BA_{jr}$ for every nonfaulty party $P_i$.
From the Validity property of $BA_{jr}$, all nonfaulty parties must output 1 in $BA_{jr}$ reaching a contradiction.

\textbf{Termination.}
First assume that all nonfaulty parties participate in the protocol, and that there exists some set $I\subseteq\left[n\right]$ such that $\left|I\right|\geq k$, and that for every nonfaulty party $P_i$ and $j\in I$ eventually $Q_{ir}\left(j\right)=1$ almost-surely. 
From lemma~\ref{lem:cs_size}, all nonfaulty parties almost-surely eventually complete at least k invocations of $BA_{jr}$ with output 1.
At that point, $c_{ir}\geq k$ holds for every nonfaulty party $P_i$.
Because of line 4, every nonfaulty party $P_i$ participates in $BA_{jr}$ for every $j\in\left[n\right]$ such that $Q_{ir}\left(j\right)=0$ at that point in time.
It is important to note that if $Q_{ir}\left(j\right)\neq 0$ then it must equal 1, which means that $P_i$ has already invoked $BA_{jr}$ with input 1 previously.
In other words, all nonfaulty parties have invoked $BA_{jr}$ for every $j\in\left[n\right]$, so from the Termination property of $BA$ they almost-surely complete all of those invocations.
At that point they reach line 6 of the protocol, and complete $CommonSubset_r$.

For the second part of the property observe some nonfaulty party $P_l$ that participates in the $CommonSubset$ protocol. 
If some nonfaulty party $P_i$ completed the $CommonSubset_r$ protocol, it must have completed the $BA_{jr}$ invocation for every $j\in\left[n\right]$.
Let $S_r$ be $P_i$'s output in this invocation of the $CommonSubset_r$ protocol.
From the Correctness property of $CommonSubset_r$, for every $j\in S_r$, $Q_{kr}\left(j\right)=1$ for some nonfaulty party $P_k$.
Since for some nonfaulty party $P_k$ $Q_{kr}\left(j\right)=1$, by assumption eventually $Q_{lr}\left(j\right)=1$ as well.
At that point, if $P_l$ hasn't started participating in $BA_{jr}$ with input $0$, it starts participating in it with input $1$.
Since $P_i$ completed each of those $BA$ invocations, from the Termination property of $BA$, $P_l$ almost-surely completes them as well.
Note that after completing the $CommonSubset_r$ invocation, all nonfaulty parties continue participating in all relevant $BA$ invocations until they terminate.
From the Correctness property of $BA$, party $P_l$ outputs $1$ in every $BA_{jr}$ invocation such that $j\in S_r$ because $P_i$ must have output $1$ in that invocation as well.
From the Correctness Property of $CommonSubset_r$, $\left|S\right|\geq k$ and thus at that point $c_{lr}\geq k$.
At that point, $P_l$ inputs $0$ to every $BA$ invocation it hasn't started participating in yet.
Following similar arguments, from the Termination property of $BA$ $P_l$ almost-surely completes all of those invocations and then completes the protocol.

\end{proof}

\section{Completion of the Proof of Theorem~\ref{thm:coinflip}}
In order to complete the proof of the Correctness property, it is left to show that $\Pr\left[\left|\left\{r|c'_r=1\right\}\right|> \frac {k}{2}+n^2\right]\geq \frac{1}{2} - \epsilon$.
If that is the case, every nonfaulty party that completes the protocol outputs 1.
Since for every $r\in\left[k\right]$, $\Pr\left[c'_r=1\right]=\frac{1}{2}=\Pr\left[c'_r=0\right]$, the case for 0 is entirely symmetric.
Define the random variable $X=\left|\left\{r|c'_r=1\right\}\right|$. Each $c'_r$ is an independent Bernoulli variable with probability $\frac{1}{2}$ of being 1, and thus $X\sim Bin\left(k,\frac{1}{2}\right)$.
In this analysis we use the fact that:
\begin{align*}
    n! & \leq e\cdot n^{n+\frac{1}{2}}\cdot e^{-n} \\ 
    n! & \geq \sqrt{2\pi}\cdot n^{n+\frac{1}{2}}\cdot e^{-n}
\end{align*}

Start by bounding the size of $\binom{2n}{n}$ for any $n$:

\begin{align*}
    \binom{2n}{n} & = \frac{\left(2n\right)!}{\left(n!\right)^2} \\ 
    & \leq \frac{e\left(2n\right)^{2n+\frac{1}{2}}e^{-2n}}
    {\left(\sqrt{2\pi}\left(n\right)^{n+\frac{1}{2}}e^{-n}\right)^2} \\
    & = \frac{e}{2\pi}\cdot\frac{\left(2n\right)^{2n+\frac{1}{2}}}{\left(n\right)^{2n+1}} \\
    & = \frac{e}{2\pi}\cdot 2^{2n+\frac{1}{2}}\cdot \frac{1}{\sqrt{n}}
\end{align*}

Denote $k=4\ceil{c^2n^4}$ with $c=\frac{e}{\epsilon\cdot\pi}$, and $\mu=\frac{k}{2}=2\ceil{c^2n^4}$. 
Now bound the probability that $X$ is very close to $\mu$:
\begin{align*}
    \Pr\left[\mu -n^2\leq X \leq \mu + n^2\right] & = \sum_{\mu -n^2\leq l \leq \mu + n^2} \binom{2\mu}{l}\left(\frac{1}{2}\right)^{2\mu} \\
    & \leq \left(2n^2+1\right) \binom{2\mu}{\mu}\left(\frac{1}{2}\right)^{2\mu} \\
    & \leq \left(2n^2+1\right) \frac{e}{2\pi}\cdot 2^{2\mu+\frac{1}{2}}\cdot \frac{1}{\sqrt{\mu}}
    \left(\frac{1}{2}\right)^{2\mu} \\
    & = \left(2n^2+1\right) \cdot \frac{e}{2\pi}\cdot \frac{1}{\sqrt{\mu}}
    \cdot\sqrt{2}
\end{align*}

Substituting back $\mu=2\ceil{c^2n^4}$:
\begin{align*}
    \Pr\left[\mu -n^2\leq X \leq \mu + n^2\right] & \leq 
    \left(2n^2+1\right) \cdot \frac{e}{2\pi}\cdot \frac{1}{\sqrt{\mu}}
    \cdot\sqrt{2} \\
    & = \left(2n^2+1\right) \cdot \frac{e}{2\pi}\cdot \frac{1}{\sqrt{2\ceil{c^2n^4}}} \cdot\sqrt{2}\\
    & \leq \left(2n^2+1\right) \cdot \frac{e}{2\pi}\cdot \frac{1}{cn^2} \\
    & = \frac{2n^2 + 1}{n^2}\cdot \frac{e}{2\pi}\cdot \frac{1}{c} \\
    & \leq \frac{2e}{\pi}\cdot\frac{1}{c}
\end{align*}

Since the cases that $X>\mu+n^2$ and $X<\mu-n^2$ are entirely symmetric:

\begin{align*}
    \Pr\left[X>\mu + n^2\right] & = \frac{1}{2}\left(1-\Pr\left[\mu -n^2\leq X \leq \mu + n^2\right]\right) \\
    & \geq \frac{1}{2}\left(1-\frac{2e}{\pi}\cdot\frac{1}{c}\right) \\
    & = \frac{1}{2} - \frac{e}{\pi}\cdot \frac{1}{c}
\end{align*}

Finally, substituting $c = \frac{e}{\epsilon\cdot\pi}$ and $\mu=\frac{k}{2}$:

\begin{align*}
    \Pr\left[X>\frac{k}{2} + n^2\right] 
    & \geq  \frac{1}{2} - \frac{e}{\pi}\cdot \frac{1}{c} \\
    & = \frac{1}{2} - \frac{e}{\pi}\cdot \frac{\epsilon\cdot\pi}{e} \\
    & = \frac{1}{2}-\epsilon
\end{align*}
which completes the proof.

\section{Proof of Theorem~\ref{thm:fairchoice}}
\fairchoice*
\begin{proof}
Each property is proven individually.
Throughout the analysis, unless explicitly stated differently all logarithms are treated as logarithms with base 2.

\textbf{Termination.}
If all nonfaulty parties participate in the protocol and have the same input $m$, they all compute the same values $l$ and $\epsilon$.
They then all participate in the $CoinFlip$ protocol $l$ times with the same parameter $\epsilon$ and from the Termination property of the $CoinFlip$ protocol, they all almost-surely complete each of those invocations.
Afterwards every nonfaulty party performs some local computations and completes the protocol.
On the other hand, if some nonfaulty party completes the $FairChoice$ protocol, it must have first completed all $l$ invocations of the $CoinFlip$ protocol with parameter $\epsilon$.
Observe some other nonfaulty party $P_i$ that participates in the $Fair Choice$ protocol with the same input $m$.
It must have computed the same values $l$ and $\epsilon$, and then participated in $l$ invocations of the $CoinFlip$ protocol with the same parameter $\epsilon$.
Since some nonfaulty party completed all $l$ of those invocations, from the Termination property of the $CoinFlip$ protocol, $P_i$ almost-surely completes them as well.
Afterwards $P_i$ performs some local computations and completes the protocol.

\textbf{Correctness.} 
Observe two nonfaulty parties that complete the protocol.
Since they both have the same input $m$, they must have computed the same value $l$, and participated in $l$ invocations of the $CoinFlip$ protocol.
From the Correctness property of the $CoinFlip$ protocol, for every $i\in\left[l\right]$ they must have output the same value $b_i\in\left\{0,1\right\}$ in the $i$'th invocation of the $CoinFlip$ protocol.
This means that they compute the same number $r$, and then both output output $r \mod m\in\left\{0,\ldots,m-1\right\}$.

\textbf{Validity.}
Intuitively, there are more values in $G$ than values not in $G$ and each value $i\in G$ has almost the same number of numbers $k\in\left[l\right]$ such that $k \equiv i \mod m$.
Furthermore, each number in $\left[l\right]$ has nearly the same probability of being sampled.
If every number had the exact same probability of being sampled, and each value $i\in G$ had exactly the same number of numbers $k\in\left[l\right]$ such that $k \equiv i \mod m$ it is clear that the property holds.
It is only left to show that these slight differences aren't big enough for the property not to hold.

Consider the case in which all nonfaulty parties that participate in the protocol have the same input $m$.
Let $N,l,\epsilon$ be defined as they are in the protocol.
Consider some $G\subseteq\left\{0,\ldots,m-1\right\}$ such that $\left|G\right|>\frac{m}{2}$.
For every $i\in\left\{0,\ldots,m-1\right\}$ define the set $S_i=\left\{j\in\left\{0,\ldots,N-1\right\}|j\equiv i \mod m\right\}$.
Define $S=\cup_{i\in G} S_i$.
First, bound the size of $S$.
For every $i\in\left\{0,\ldots,m-1\right\}$, 
$\left|S_i\right| \geq \lfloor \frac{N}{m} \rfloor \geq \frac{N}{m}-1$.

Since $\left|G\right|,m\in\mathbb{N}$:
\begin{align*}
    \left|G\right| &> \frac{m}{2} \\
    2\left|G\right| &> m \\
    2\left|G\right| &\geq m+1 \\
    \left|G\right| &\geq \frac{m}{2}+\frac{1}{2} \\
\end{align*}
Note that for every $i\neq j$ $S_i\cap S_j=\emptyset$ and thus:
\begin{align*}
    \left|S\right| &=\sum_{i\in G} \left|S_i\right|\\
    &\geq \left(\frac{N}{m}-1\right)\left|G\right|\\
    &\geq\left(\frac{N}{m}-1\right)\left(\frac{m}{2}+\frac{1}{2}\right) \\
    &=\left(N-m\right)\left(\frac{1}{2}+\frac{1}{2m}\right)
\end{align*}

As shown in the proof of the Correctness property, all nonfaulty parties that complete the protocol first complete $l$ invocations of the $CoinFlip$ protocol, output the same bits $b_i$ for every $i\in\left[l\right]$, then compute the same value $r$ and output $r\mod m$.
In that case, all nonfaulty parties output some $i\in G$ if and only if $r\in S$.
From the Correctness property of the $CoinFlip$ protocol, for every $j\in\left[l\right]$ and $b\in\left\{0,1\right\}$, $\Pr\left[b_j=b\right]\geq \frac{1}{2}-\epsilon$ regardless of the adversary's actions.
For every number $r$ denote $r_i$ to be the $i'th$ bit in its binary representation.
Therefore:
\begin{align*}
    \Pr\left[i\in G\right] & = \Pr\left[r\in S\right] \\
    & =\sum_{r'\in S} \Pr\left[r=r'\right] \\
    & =\sum_{r'\in S} \Pr\left[\bigwedge_{j=1}^l r_j=r'_j\right] \\
    & \geq \sum_{r'\in S} \left(\frac{1}{2}-\epsilon\right)^l \\
    & =\left|S\right|\left(\frac{1}{2}-\epsilon\right)^l \\
    & \geq \left(N-m\right)
    \left(\frac{1}{2}+\frac{1}{2m}\right)
    \left(\frac{1}{2}-\epsilon\right)^{\log N} \\
    & = \left(N-m\right)
    \left(\frac{1}{2}+\frac{1}{2m}\right)
    \left(\frac{1}{2}\right)^{\log N}
    \left(1-2\epsilon\right)^{\log N} \\
    & = \left(1-\frac{m}{N}\right)
    \left(\frac{1}{2}+\frac{1}{2m}\right)
    \left(1-\frac{2}{100m\log m}\right)^{\log N} \\
    &\geq \left(1-\frac{m}{2m^2}\right)
    \left(\frac{1}{2}+\frac{1}{2m}\right)
    \left(1-\frac{1}{50m\log m}\right)^{\log 4m^2} \\
    & = \left(\frac{1}{2}+\frac{1}{2m}-\frac{1}{4m}-\frac{1}{4m^2}\right) \\
    & \cdot\left(\left(1-\frac{1}{50m\log m}\right)^{m\log m}\right)^{\frac{2\log m+2 }{m\log m}}
\end{align*}

At this point recall that $m\geq 3$.
First, clearly $\frac{2\log m +2}{m\log m}\leq \frac{4\log m }{m\log m}=\frac{4}{m}$ for any $m\geq 2$.
Secondly, note that the expression $\left(1-\frac{x}{n}\right)^n$ approaches $e^{-x}$ from below in a monotonously increasing manner for $1>x>0$.
Plugging in $m=3$, $\left(1-\frac{1}{50\cdot 3\log 3}\right)^{3\log 3}\geq \frac{99}{100}e^{-\frac{1}{50}}$, and from the previous observation this means $\left(1-\frac{1}{50 m\log m}\right)^{m\log m}\geq \frac{99}{100}e^{-\frac{1}{50}}$ for every $m\geq 3$.
Combining these observations:
\begin{align*}
    &\Pr\left[i\in G\right] \geq\\ 
    & \geq \left(\frac{1}{2}+\frac{1}{2m}-\frac{1}{4m}-\frac{1}{4m^2}\right)
    \left(\left(1-\frac{1}{50m\log m}\right)^{m\log m}\right)^{\frac{2\log m+2 }{m\log m}} \\
    & \geq \left(\frac{1}{2}+\frac{1}{4m}-\frac{1}{4m^2}\right)\left(\frac{99}{100}e^{-\frac{1}{50}}\right)^\frac{4}{m}
\end{align*}
First, note that clearly:
\begin{align*}
    \lim_{m\to\infty}\left(\frac{1}{2}+\frac{1}{4m}-\frac{1}{4m^2}\right)\left(\frac{99}{100}e^{-\frac{1}{50}}\right)^\frac{4}{m}=\left(\frac{1}{2}\right)\left(1\right)=\frac{1}{2}
\end{align*}
In addition, setting $m=3$ and checking numerically:
\begin{align*}
    \left(\frac{1}{2}+\frac{1}{4\cdot3}-\frac{1}{4\cdot 3^2}\right)\left(\frac{99}{100}e^{-\frac{1}{50}}\right)^\frac{4}{3}\approx 0.534 > 0.5
\end{align*}
Next observe the derivative of the expression with respect to $m$ and check when it is negative.
\begin{align*}
    &\frac{d}{d m}\left(\frac{1}{2}+\frac{1}{4m}-\frac{1}{4m^2}\right)
    \left(\frac{99}{100}e^{-\frac{1}{50}}\right)^{\frac{4}{m}} =\\
    &=\left(-\frac{1}{4m^2}+\frac{1}{2m^3}\right)
    \left(\frac{99}{100}e^{-\frac{1}{50}}\right)^{\frac{4}{m}} \\
    &+ \left(\frac{1}{2}+\frac{1}{4m}-\frac{1}{4m^2}\right)
    \left(\frac{99}{100}e^{-\frac{1}{50}}\right)^{\frac{4}{m}}
    \left(\frac{-4\ln\left(\frac{99}{100}e^{-\frac{1}{50}}\right)}{m^2}\right) \\
    & = \left(\frac{99}{100}e^{-\frac{1}{50}}\right)^{\frac{4}{m}}
    \left(\frac{2-m}{4m^3}-
    \frac{16m\left(\frac{1}{2}+\frac{1}{4m}-\frac{1}{4m^2}\right)\ln\left(\frac{99}{100}e^{-\frac{1}{50}}\right)}{4m^3}\right) \\
    & = \frac{1}{4m^3}
    \left(\frac{99}{100}e^{-\frac{1}{50}}\right)^{\frac{4}{m}}
    \left(2-m-\left(8m+4-\frac{4}{m}\right)\ln\left(\frac{99}{100}e^{-\frac{1}{50}}\right)\right)
\end{align*}
Now note that for any $m\geq 3$:
\begin{align*}
    \frac{1}{4m^3}
    \left(\frac{99}{100}e^{-\frac{1}{50}}\right)^{\frac{4}{m}} > 0
\end{align*}
and thus the whole expression is negative if:
\begin{align*}
    0 &> 2-m-\left(8m+4-\frac{4}{m}\right)\ln\left(\frac{99}{100}e^{-\frac{1}{50}}\right) \\
    & = 2-m+\left(8m+4-\frac{4}{m}\right)
    \ln\left(\frac{100}{99}e^{\frac{1}{50}}\right)
\end{align*}
Numerically we can find that $0.031\geq\ln\left(\frac{100}{99}e^{\frac{1}{50}}\right)>0$ and thus:
\begin{align*}
    & 2-m+\left(8m+4-\frac{4}{m}\right)
    \ln\left(\frac{100}{99}e^{\frac{1}{50}}\right) \leq\\
    &\leq 2-m+\left(8m+4\right)
    \ln\left(\frac{100}{99}e^{\frac{1}{50}}\right) \\
    &\leq 2-m+\left(8m+4\right)0.031 \\
    & = 2-m+0.248m+0.124\\
    & = 2.124-0.752m
\end{align*}
Finally check if this term is negative:
\begin{align*}
    2.124-0.752m&<0\\
    2.124&<0.752m \\
    \frac{2.124}{0.752}\approx 2.824 &< m
\end{align*}
Since $m\geq 3$:
\begin{align*}
    2-m+\left(8m+4-\frac{4}{m}\right)
    \ln\left(\frac{100}{99}e^{\frac{1}{50}}\right)<0
\end{align*}
and thus for every $m\geq 3$:
\begin{align*}
    \frac{d}{d m}\left(\frac{1}{2}+\frac{1}{4m}-\frac{1}{4m^2}\right)
    \left(\frac{99}{100}e^{-\frac{1}{50}}\right)^{\frac{4}{m}}
    < 0
\end{align*}

Combining the fact that at $m=3$ the expression is greater than $\frac{1}{2}$, that the derivative is negative for any $m\geq 3$ and that the expression approaches $\frac{1}{2}$ as $m$ approaches infinity, for any $m\geq 3$:

\begin{align*}
    \Pr\left[i\in G\right] & \geq \left(\frac{1}{2}+\frac{1}{4m}-\frac{1}{4m^2}\right)\left(\frac{99}{100}e^{-\frac{1}{50}}\right)^\frac{4}{m} > \frac{1}{2}
\end{align*}
\end{proof}

\section{Proof of Theorem~\ref{thm:fba}}
\fba*
\begin{proof}
Again, each property is proven individually.

\textbf{Termination.}
If all nonfaulty parties participate in the $FBA$ protocol, they all A-Cast some values in step 1 and participate in each other's A-Casts.
Since all of the senders in those A-Casts are nonfaulty and all nonfaulty parties participate in all of those A-Casts, from the Termination property of A-Cast they all complete each of those invocations.
This means that for every pair of nonfaulty parties $P_i,P_j$ eventually $Q_i\left(j\right)=1$.
In other words, since there are at least $n-t$ nonfaulty parties there exists a set $I\subseteq\left[n\right]$ such that for every nonfaulty party $P_i$, eventually $\forall j\in I \ Q_i\left(j\right)=1$.
From the Termination property of $CommonSubset$, all nonfaulty parties almost-surely eventually complete the protocol.
From the Correctness property of $CommonSubset$, all nonfaulty parties output the same $S\subseteq\left[n\right]$ and for every $j\in S$ there exists some nonfaulty party $P_i$ such that $Q_i\left(j\right)=1$.
A nonfaulty party sets $Q_i\left(j\right)=1$ only if it completed $P_j$'s A-Cast, and from the Termination property of A-Cast, all nonfaulty parties that participate in that A-Cast complete it as well.
This means that all nonfaulty parties complete all relevant A-Cast invocations and then finish step 4 of the protocol.
From the Correctness property of A-Cast, all nonfaulty parties receive the same value $x'_k$ in $P_k$'s A-Cast for every $k\in S$.
If some nonfaulty party completes the protocol in step 5, then there exists some $x$ such that $\left|\left\{x'_j=x|j\in S\right\}\right| > \frac{m}{2}$.
Every other nonfaulty party sees that this holds as well and completes the protocol in step 5.
Otherwise, all nonfaulty parties participate in $FairChoice\left(m\right)$ with the same $m=\left|S\right|$.
Note that from the correctness property of the $CommonSubset$ protocol, they all output some set $S$ such that $\left|S\right|\geq n-t\geq 3$.
From the Termination property of the $FairChoice$ protocol, they all almost-surely complete the $FairChoice$ protocol as well.
Afterwards they perform some local computations and complete the protocol.

For the second part of the property, assume some nonfaulty party $P_i$ completed the $FBA$ protocol.
This means it must have completed both the $CommonSubset$ protocol and the $FairChoice$ protocol if it didn't complete the protocol in step 5.
Observe some other nonfaulty party $P_j$ that participates in the protocol.
First, $P_j$ A-Casts some value and participates in every other party's A-Cast.
$P_j$ then participates in the $CommonSubset$ protocol.
Note that every nonfaulty party that participates in the $CommonSubset$ protocol also participates in each of the A-Cast invocations.
For every pair of nonfaulty parties $P_k,P_l$ that participate in $CommonSubset$ and value $m\in\left[n\right]$, if $Q_{k}\left(m\right)=1$ $P_k$ must have completed $P_m$'s A-Cast.
From the Termination property of A-Cast, $P_l$ will eventually complete $P_m$'s too A-Cast and set $Q_{l}\left(m\right)=1$.
Therefore, the conditions of the second part of the Termination property of the $CommonSubset$ protocol hold, and thus since $P_i$ completed the $CommonSubset$ protocol $P_j$ almost-surely completes it as well with some output $S$.
From the Correctness property of $CommonSubset$, for every $k\in S$, for some nonfaulty party $P_l$, $Q_{l}\left(k\right)=1$.
This means that $P_l$ completed $P_k$'s A-Cast, which means $P_j$ does so as well.
If $P_j$ completes the protocol in step $5$ after completing all of the A-Cast invocations we are done.
Otherwise, after completing all of the relevant A-Casts, $P_j$ participates in the $FairChoice$ protocol. 
In that case there must not exist any $x$ such that $\left|\left\{x'_k=x|k\in S\right\}\right|>\frac{m}{2}$.
From the Correctness property of A-Cast, $P_i$ must have output the same value in each of those A-Casts, seen that there does not exist any $x$ such that $\left|\left\{x'_k=x|k\in S\right\}\right|>\frac{m}{2}$, and then invoked and completed the $FairChoice$ protocol.
From the Termination property of the $FairChoice$ protocol, $P_j$ almost-surely completes it as well, performs some local computations, and then finally completes the $FBA$ protocol.

\textbf{Correctness.} 
Let $P_i,P_j$ be two nonfaulty parties that completed the protocol.
They must have both first completed the $Common\-Subset$ protocol and from its Correctness property output the same set $S\subseteq\left[n\right]$.
They then completed $P_k$'s A-Cast for every $k\in S$.
From the Correctness property of A-Cast, they both received the same value $x'_k$ for every $k\in S$.
If there exists some $x$ such that $\left|\left\{x'_k=x|k\in S\right\}\right| > \frac{m}{2}$ then they both must have output that value and completed the protocol.
Note that clearly there cannot be more than one such value.
If there isn't any such value $x$, then they both participated in the $FairChoice$ protocol and from the Correctness property of the protocol output the same value $k\in\left\{0,\ldots,m-1\right\}$.
They then both took the $k$'th biggest value in $S$ and output the value corresponding to that party's A-Cast.
Again, from the Correctness property of A-Cast $P_i,P_j$ must have received the same value in that A-Cast and thus output the same value.

\textbf{Validity.} 
First assume that all nonfaulty parties have the same input $x$.
In that case, in the beginning of the protocol each nonfaulty party that participates in the protocol  A-Casts $x$.
Let $P_i$ be some nonfaulty party that completed the protocol.
It must have first participated in all relevant A-Casts and in the $CommonSubset$ protocol, and completed the $CommonSubset$ protocol with some output $S$.
From the Correctness property of $CommonSubset$, $\left|S\right|\geq n-t$.
$P_i$ then completed $P_j$'s A-Cast for every $j\in S$.
From the Validity property of A-Cast, for every nonfaulty party $P_j$ such that $j\in S$, $P_i$ received the value $x'_j=x$.
Let $G$ be the set of all $j\in S$ such that $P_j$ is nonfaulty.
Since there are at most $t$ faulty parties $P_k$ such that $k\in S$, $\left|G\right|\geq \left|S\right|-t=m-t$.
Note that $m\geq n-t > 2t$ and thus:
\begin{align*}
    m&>2t \\
    \frac{m}{2}&>t \\
    m-t&>\frac{m}{2}
\end{align*}
Since $\left|G\right|\geq m-t>\frac{m}{2}$, and for every $j\in G$, $P_i$ received the value $x'_j=x$, $P_i$ sees that $\left|\left\{x'_j=x|j\in S\right\}\right|>\frac{m}{2}$.
This means that in step 5 $P_i$ outputs $x$ and completes the protocol.

On the other hand, if it is not the case that all nonfaulty parties had the same input, for every nonfaulty party $P_j$ let $x_j$ be its input.
Observe some nonfaulty party $P_i$ that completed the protocol.
Following the exact same arguments as above, $P_i$ must have participated in all A-Casts, completed the $CommonSubset$ protocol with some output $S$ such that $m=\left|S\right|\geq n-t$, and completed $P_j$'s A-Cast for every $j\in S$.
Note that from the Validity property of A-Cast, for every nonfaulty party $P_j$, $P_i$ received the value $x_j=x'_j$ in $P_j$'s A-Cast.
If $P_i$ output some value in step 5, it must have found some value $x$ such that $\left|\left\{x'_j=x|j\in S\right\}\right|>\frac{m}{2}$.
As previously shown $\frac{m}{2}>t$, and thus $\left|\left\{x'_j=x|j\in S\right\}\right|\geq t+1$.
There are $t$ faulty parties at most, which means that there must be some nonfaulty party $P_j$ such that $x_j=x'_j=x$.
In other words, if $P_i$ completed the protocol in step 5 the property holds.
Otherwise, $P_i$ must have invoked and completed protocol $FairChoice$ before completing the $FBA$ protocol.
Define $G$ as defined above.
As previously shown $\left|G\right|>\frac{m}{2}$.
Let $S_G\subseteq \left\{0,\ldots,m-1\right\}$ be all of the numbers $k\in\left\{0,\ldots,m-1\right\}$ such that the $k$'th biggest value in $S$ (as defined in the protocol) is in $G$.
Note that each $k\in\left\{0,\ldots,m-1\right\}$ corresponds to a unique value $j\in S$, and thus $\left|S_G\right|=\left|G\right|>\frac{m}{2}$.
From the Correctness property of $FairChoice$, with probability $\frac{1}{2}$ at the very least $P_i$ outputs some $k\in S_G$.
$P_i$ then finds the corresponding $j\in G$ and outputs $x'_j=x_j$.
Since by definition $P_j$ is a nonfaulty party, $P_i$ output some nonfaulty party's input, completing the proof.
\gs{This whole argument is pretty ugly. I'm inclined to be much less formal and just say why it's true (even though $S_G$ isn't formally defined here either).}
\end{proof}
\end{document}